\documentclass[12pt]{amsart}

\usepackage{amsmath,amssymb,latexsym}
\usepackage{fullpage}
\usepackage{etoolbox}
\usepackage{enumitem}
\usepackage{color}
\usepackage[colorlinks,linkcolor=blue,citecolor=magenta]{hyperref}
\usepackage{soul}
\usepackage{microtype}
\usepackage{yhmath}
\usepackage{graphicx}
\usepackage{mathtools}
\usepackage{array}
\usepackage{float}
\usepackage{complexity}
\usepackage[sharp]{easylist}
\usepackage{bbm}
\usepackage{tabularx,lipsum,environ}
\usepackage[ruled, linesnumbered]{algorithm2e}
\usepackage{subcaption}

\newtheorem{theorem}{Theorem}[section]
\newtheorem{proposition}[theorem]{Proposition}

\newtheorem{lemma}[theorem]{Lemma}

\theoremstyle{definition}

\renewcommand{\geq}{\geqslant}
\renewcommand{\leq}{\leqslant}

\def\Z{\mathbb{Z}}
\def\R{\mathbb{R}}
\def\calF{\mathcal{F}}

\def\calS{\mathcal{S}}
\def\calI{\mathcal{I}}
\def\llambda{\boldsymbol{\lambda}}
\def\un{\boldsymbol{1}}

\title{Balanced assignments of periodic tasks}

\author{Héloïse Gachet}
\address{H. Gachet, CERMICS, ENPC, Institut Polytechnique de Paris, Marne-la-Vallée \& SNCF, DTPIG, France}
\email{heloise.gachet@enpc.fr}

\author{Frédéric Meunier}
\address{F. Meunier, CERMICS, ENPC, Institut Polytechnique de Paris, Marne-la-Vallée, France}
\email{frederic.meunier@enpc.fr}
\thanks{A preliminary version of this paper appeared in the Proceedings of the 24th Symposium on Algorithmic Approaches for Transportation Modelling, Optimization, and Systems (ATMOS 2024).}

\begin{document}

\begin{abstract}
This work addresses the problem of assigning periodic tasks to workers in a {\em balanced} way, i.e., so that each worker performs every task with the same frequency over the long term. The input consists of a list of tasks to be repeated weekly at fixed times and a number of indistinguishable workers. In the basic version, the sole constraint is that no worker performs two tasks simultaneously. In the extended version, additional constraints can be introduced, such as limits on the total number of working hours per week.

Regarding the basic version, a necessary and sufficient condition for the existence of a balanced assignment is established. This condition can be verified in polynomial time. For the extended version, it is demonstrated that whenever a balanced assignment exists, a {\em periodic} balanced assignment exists as well, with a tighter bound on the period for the basic version.
\end{abstract}


%\subjclass[2012]{G.2.3}

\maketitle

\section{Introduction}
\subsection{Context}
Assigning periodic tasks to workers gives rise to problems widely studied in the Operations Research literature. Korst, Aarts, Lenstra, and Wessels~\cite{korst94} consider for instance the problem of assigning periodic operations to processors, while minimizing the total number of such processors. 
% In particular, they show the close link between the latter problem and the coloring of a circular-arc graph, which has been proven to be NP-hard by Garey, Johnson, Miller, and Papadimitriou~\cite{garey80}. A full body of research is devoted to the study of circular-arc graphs~\cite{tucker75, teng85}.
As another example, Orlin~\cite{orlin82} studies the problem of minimizing the total number of airplanes while meeting a fixed periodically repeating set of tasks. 
The question of the periodicity of such assignments is also often addressed~\cite{PespAtmos2024}.
% Serafini and Ukovich~\cite{serafini1989}

% {Ajouter brièvement au moins deux autres citations, anciennes et fondamentales}

% {\blue + 2 articles et "autres articles + récents qui mentionnent la périodicité des affectations" (copains ATMOS)}
% The question of the periodicity of assignments is often addressed as Bortoletto, Van Lieshout, Masing and Linder do in \cite{PespAtmos2024}. They consider a periodic timetabling problem, assigning activities to infrastructural elements such as platforms and obtain periodicity results on the assignments.

In this article, we consider tasks that need to be repeated every week at specific times and a number of indistinguishable workers to perform them. We address the problem of assigning each occurrence of these tasks to the workers in a {\em balanced} way, i.e., so that each worker performs each task with the same frequency over a long period. Two versions of this problem are considered: the basic version, for which the only constraint is that a worker cannot perform two tasks at the same time, and the extended version where additional constraints on the feasibility of weekly schedules are considered. As an illustration, the latter version covers the case of a limit on the total number of working hours per week. %{\blue Even in the basic version of this problem, there are cases where there is no balanced feasible, even when a feasible assignment exists (an illustration is given in Figure~\ref{fig:unachievable}).}

Based on the authors' experience, this concern of balancedness between workers is present in multiple industries, especially in the transportation sector. For example, at the SNCF, the main French railway company, the schedules of the freight train drivers must be of a certain form, called ``cyclic rosters'' in the literature \cite{breugem20, xie15}, that automatically ensures such balancedness while helping to maintain a similar level of proficiency among the workers. However, we have not been able to find any paper addressing the question of achieving balanced assignment of tasks in the long term as we do in this article.
%As another example where such fairness is sought, Meyer~\cite{equitable_coloring} introduces the equitable coloring problem, later widely studied in the literature, consisting of finding a proper coloring in a graph such that the sizes of the color sets differ by at most one. The original motivation for this problem is given by the application to the garbage collection routing problem, but could also be relevant for assigning tasks to workers in an almost balanced manner.

Our work shows that, for the two versions of the studied problem, whenever there exists a balanced assignment, there exists one that is periodic. We prove that, for the basic version, the period can be made equal to the number of workers. Furthermore, still for the basic version, we show that the existence of such an assignment is polynomial time decidable. As a tool to establish our results, we introduce a problem of pebbles moving along arcs on an arc-colored directed graph and that need to visit each arc with the same frequency in the long term. This latter problem and the results obtained about it might be of independent interest.

%{\blue Despite the fact that this problem is very natural, this problem has not been studied in the literature, as far as o the authors knowledge.}

\subsection{Problem formulation and first results}\label{sec:intro}

Consider a collection of tasks that have to be performed periodically (typically every week). We assume that each task has a given starting time and a given ending time over the period and that its length is at most the period. These tasks have to be performed by a group of indistinguishable workers. Formally, we are given an instance $\calI$ with
\begin{itemize}
    \item a collection of $n$ sub-intervals $[s_i,e_i)$ of the open interval $(-1,1)$, with $e_i \in (0,1]$ and $e_i - s_i \leq 1$.
    \item a positive integer $q$.
\end{itemize}
Each interval $[s_i,e_i)$ represents a task: the $r$th occurrence of task $i$ ($r\in \Z_{>0}$) takes place over the time interval $[s_i+r,e_i+r)$. The number $q$ corresponds to the number of workers, whom we identify from now on with the set $[q]$. 
The interval $(r,r+1]$ will be referred to as {\em week} $r$. The tasks that are overlapping the left endpoint of the interval $[0,1]$ will play an important role in the proofs and, for later purpose, we denote by $U(\calI)$ the set of such tasks. Formally, $U(\calI) = \{i\in [n] \colon s_i \leq 0\}$.

Every occurrence of each task has to be assigned to a worker. Such an assignment is {\em feasible} if it satisfies the {\em non-overlapping property}:  no worker is assigned two occurrences overlapping within $\R_{>0}$. Such an assignment is {\em balanced} if each task is performed by each worker every $q$ weeks in the long term average.
The notion of feasibility of an assignment will be revised in Section~\ref{sec:extension}, where the extended version will be defined. Until then, we refer to the basic version of the problem, where feasibility deals only with the non-overlapping property. 

In symbols, consider a function $f \colon [n] \times \Z_{>0} \to [q]$, where $f(i,r) = j$ means that the $r$th occurrence of task $i$ is assigned to worker $j$. The non-overlapping property writes then:
\begin{equation}\label{eq:feasible}
[s_i+r,e_i+r) \cap [s_{i'}+r',e_{i'}+r') \neq \varnothing \quad \Longrightarrow \quad f(i,r) \neq f(i',r')
\end{equation}
for all $i \neq i'$ and all $r,r'$. (Remark that the left-hand side holds only if $|r-r'|\leq 1$ and that two occurrences of the same task never overlap.)
It is balanced if
\begin{equation}\label{eq:limit}
\lim_{t\to +\infty} \frac 1 t \big|\{r \in [t] \colon f(i,r) = j\}\big| = \frac 1 q 
\end{equation}
for all $i \in [n]$ and all $j \in [q]$. An illustration is given in Figures~\ref{fig:unachievable} and~\ref{fig:ex}, the first figure providing an example with a feasible assignment but no balanced feasible assignment. An assignment $f$ is {\em periodic} if there exists $h \in \Z_{>0}$  such that $f(i,r) = f(i,r+h)$ for all $i\in [n]$ and $r\in\Z_{>0}$.

\begin{figure}
    \begin{subfigure}{\textwidth}
    \centering
    \includegraphics[width = 0.7\textwidth]{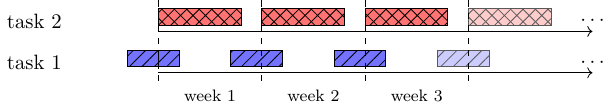}
    \caption{The two tasks of $\calI$. The first four occurrences of each task are represented. The fourth occurrence of each task is represented in lighter color. Here, $U(\calI)=\{1,2\}$. }
    \end{subfigure}

    \vspace{2mm}
    \begin{subfigure}{\textwidth}
    \centering
    \includegraphics[width = 0.7\textwidth]{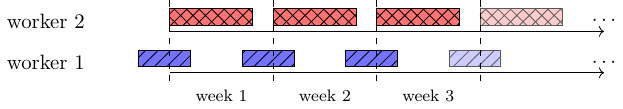}
    \caption{A feasible assignment $f$ for $\calI$. Assuming that this pattern is repeated along the horizontal axis, each line represents a worker and the $r$th occurrence of task $i$ is on the line of worker $j$ when $f(i,r)=j$. }
    \end{subfigure}
    \caption{Example of an instance $\calI$ with two tasks ($n=2$) and two workers ($q=2$), with a feasible assignment. Up to exchanging workers $1$ and $2$, this assignment is unique and it is not balanced: without loss of generality, the first occurrence of the hatched blue task is assigned to worker $1$, and this determines completely the assignment of all other occurrences of tasks $1$ and $2$.
    }
    \label{fig:unachievable}
\end{figure}

\begin{figure}
    \centering
    \begin{subfigure}{\textwidth}
    \centering
    \includegraphics[width = 0.7\textwidth]{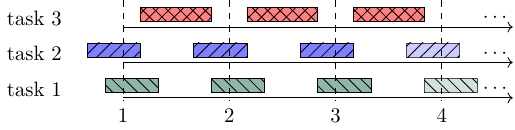}
    \caption{The three tasks of $\calI$. The first three occurrences of each task are represented. The fourth occurrence of tasks 1 and 2 is represented in lighter color. Here, $U(\calI)=\{1,2\}$.}
    \label{fig:instance}
    \end{subfigure}
    
    \vspace{2mm}
    \begin{subfigure}{\textwidth}
    \centering
    \includegraphics[width = 0.7\textwidth]{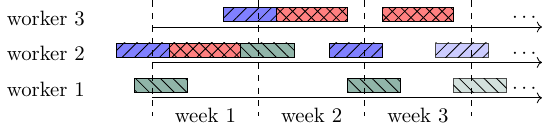}
    \caption{A feasible assignment $f$ for $\calI$. Assuming that this pattern is repeated along the horizontal axis, the assignment $f$ is not balanced: worker 3 does not perform task 1 and worker 1 performs 75\% of the occurrences of task 1.}
    \label{fig:feasible}
    \end{subfigure}
    
    \vspace{2mm}
    \begin{subfigure}{\textwidth}
    \centering
    \includegraphics[width = 0.7\textwidth]{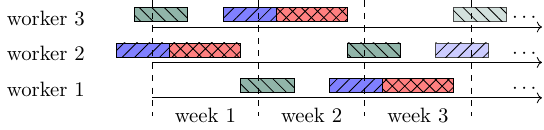}
    \caption{\label{fig:balanced}A feasible assignment $g$ for $\calI$. Assuming that this pattern is repeated along the horizontal axis, the assignment $g$ is feasible and balanced: task 1 is performed equally by workers 1, 2, and 3, and so are tasks 2 and 3. The assignment $g$ is periodic with period $h=3$.}
    
    \end{subfigure}
    \caption{\label{fig:ex} Example of an instance $\calI$ with three tasks ($n=3$) and three workers ($q=3$), with two feasible assignments.}
\end{figure}

We aim at identifying conditions under which there exists a balanced feasible assignment and at studying the related algorithmic questions.

A few comments are in order. First, note that there exists a feasible assignment if and only if there is no point in $\R$ contained in more than $q$ intervals $[s_i+r,e_i+r)$. This is an elementary result from the coloring of interval graphs. (We actually state it as a proper lemma---Lemma~\ref{lem:interval_coloring}---for sake of completeness.) This means that for our problem, feasibility is not the challenge. Second, there are assignments for which the limit in \eqref{eq:limit} is not well-defined. By definition, if the limit is not well-defined, then the assignment is not balanced.
 
Clearly, a necessary condition for the existence of a balanced feasible assignment is that there is a feasible assignment in which a worker performs each task at least once. Surprisingly this condition is actually sufficient, as stated by our first theorem. 

\begin{theorem}\label{thm:cns}
    The following statements are equivalent:
    \begin{enumerate}[label=\textup{(\alph*)}]
        \item There exists a feasible assignment with a worker performing each task at least once. \label{cond1}
        \item There exists a balanced feasible assignment. \label{cond2}
        \item There exists a periodic balanced feasible assignment with period $q$. \label{cond3}
    \end{enumerate}
\end{theorem}

Note that the implications $\ref{cond2}\Rightarrow\ref{cond1}$ (discussed above) and $\ref{cond3}\Rightarrow\ref{cond2}$ are immediate. The proof focuses on showing $\ref{cond1}\Rightarrow\ref{cond3}$.

Statement~\ref{cond1} of Theorem~\ref{thm:cns} is simple enough to obtain an algorithmic counterpart. 

\begin{theorem}\label{thm:algo}
    Deciding whether there exists a balanced feasible assignment can be done in polynomial time. Moreover, when such an assignment exists, a periodic one with period $q$ can be computed in polynomial time as well.
\end{theorem}

Note that, for the computation of the balanced feasible assignment, we are not allowed to have an output of size $\Omega(q)$ because it is the number of workers that is given in input, and not the list of them. %Moreover, given that an assignment can be computed in polynomial time, then it is easy to provide to each worker the tasks to be performed every week in polynomial time: to provide the tasks to be performed by a worker $j$ on a week $r$, a checking of $f^\star(i,r)=j$ for all $i\in [n]$ is enough.

A consequence of this result is that the existence of a balanced feasible assignment is decidable, i.e., there is an algorithm that decides in finite time whether a balanced feasible assignment exists for any input. It is not clear that there is a simpler proof of this fact.
%Unfortunately, this theorem does not necessarily provide us with a polynomial construction of a periodic balanced feasible assignment of period $q$ (even though Theorem~\ref{thm:cns} holds) when the number of workers is constant. The period of the periodic balanced feasible assignment built in polynomial time when the number of workers is constant can however be bounded by $q^2q!$.

When there is a point of $[0,1)$ contained in no interval $[s_i+r,e_i+r)$ with $r\in\{ 0,1\}$, then the theorems are trivial: without loss of generality, this point is $0$, and any feasible assignment $f$ and any cyclic permutation $\pi$ of $[q]$ (without fixed points) provide a balanced feasible assignment $g$, periodic with period $q$, defined by $g(i,r) \coloneqq (\pi^r\circ f)(i,1)$ for $i\in [n]$ and $r\in \Z_{>0}$ (where $\pi^r = \underbrace{\pi\circ \cdots\circ\pi}_{r\text{ times}}$).

\subsection{Extended version: valid weekly schedules}\label{sec:extension}
We introduce now an extended version of the previous problem, relying on the notion of ``schedules.'' 

An assignment \emph{induces} \emph{schedules} in the following manner: for each worker and each week $r$, a schedule is defined by the set of tasks performed over the week $r$, possibly with a task that has started over week $r-1$ and possibly with a task that will finish over week $r+1$.
Formally, a \emph{schedule} is a pair $(T,W)$ with $T \subseteq [n]$ and $W = \varnothing$ or $W=\{i\}$ with $i\in U(\calI)$.
The sets $T$ and $W$ can respectively be interpreted as the set of tasks finishing over week $r$ and as a potential task starting over week $r$ and finishing on week $r+1$.
A schedule is \emph{non-overlapping} if for all $i' \in T$ we have
\begin{equation}\label{eq:nonover-sch}
[s_{i'}, e_{i'})\cap[s_{i''}, e_{i''}) = \varnothing \text{ for all $i''\in T\setminus\{i'\}$}\  \text{ and }\  [s_{i'}, e_{i'})\cap[s_i+1, e_i+1) = \varnothing \text{ for all $i\in W$}\, .
\end{equation}

We extend the previous problem by constraining further the notion of feasible assignment. An instance $\calI$ comes now with a set $\calS$ of non-overlapping schedules.
 Formally, an instance $\calI$ is a triple of the form $\left(\bigl([s_i,e_i)\bigl)_{i\in[n]}, q, \calS\right)$ where every schedule in $\calS$ is non-overlapping. An assignment is {\em feasible} if all the schedules it induces belong to $\calS$. As an illustration of the relevance of this extended version, it is straightforward to impose constraints such that
 \begin{itemize}
     \item no worker works more than $H$ hours per week.
     \item no worker performs more than $M$ tasks per week.
     \item no worker starts two tasks indexed with a prime number in the same week.   
 \end{itemize}

Remark that the basic version is actually the special case of the extended version when $\calS$ is formed by all possible non-overlapping schedules. (The non-overlapping property is the only constraint to take into account in the basic version.)

An implication similar as $\ref{cond2}\Rightarrow\ref{cond3}$ in Theorem~\ref{thm:cns} holds for the extended version.

\begin{theorem}\label{thm:periodic}
    Consider an instance $\calI$ with an arbitrary set $\calS$ of non-overlapping schedules.
    If there exists a balanced feasible assignment, then there exists one that is periodic. Moreover, in such a case, the periodic assignment can be chosen so that the period is at most $q^2 q!$.
\end{theorem}

As for Theorem~\ref{thm:algo}, a consequence of this result is that the existence of a balanced feasible assignment for the extended version is decidable.

However, some properties established for the basic version do not hold for the extended version in generality. First, feasibility is not equivalent anymore to having no point in $\R$ contained in $q$ tasks simultaneously. Second, according to Theorem~\ref{thm:cns}, in the basic version, as soon as there exists a feasible assignment with a worker performing each task at least once, there exists a balanced feasible assignment. Such statement does not hold in full generality for the extended version, as illustrated in Figure~\ref{fig:counter-ex}.

\begin{figure}
    \centering
    \begin{subfigure}{\textwidth}
    \centering
    \includegraphics[width = 0.5\textwidth]{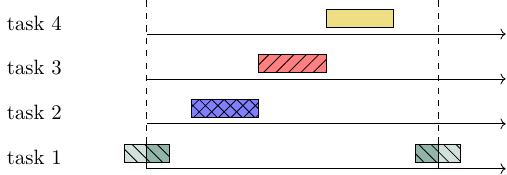}
    \caption{The four tasks of $\calI$. The first occurrence of tasks 2, 3, and 4 and the first two occurrences of task 1 are represented. Task 1 lasts 12 hours while tasks 2, 3, and 4 last 18 hours.}
    \end{subfigure}
    
    \vspace{5mm}
    \begin{subfigure}{\textwidth}
    \centering
    \includegraphics[width = \textwidth]{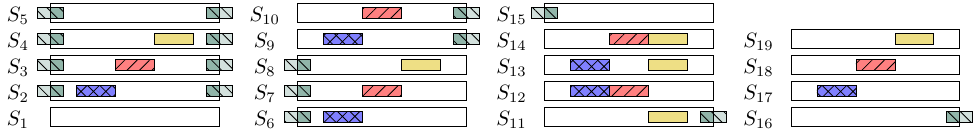}
    \caption{The set $\calS = \{S_k \colon k\in [19]\}$ gathering all non-overlapping schedules with maximum weekly work duration of 36 hours.}
    \end{subfigure}
    
    \vspace{5mm}
    \begin{subfigure}{\textwidth}
    \centering
    \includegraphics[width = \textwidth]{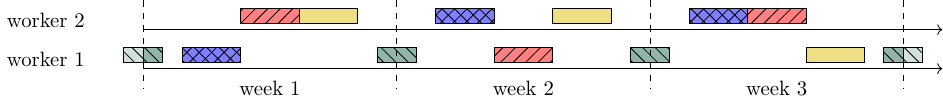}
    \caption{A feasible assignment for $\calI$. This assignment is such that worker 1 performs each task at least once; however no balanced feasible assignment exists: an easy case-checking shows that the worker starting with the first occurrence of task $1$ performs all the occurrences of this task.}
    \end{subfigure}
    \caption{An instance $\calI$ with four tasks ($n=4$), a set of non-overlapping schedules $\calS$, and two workers ($q=2$) showing that the implication $\ref{cond1}\Rightarrow \ref{cond2}$ given in Theorem~\ref{thm:cns} does not hold for the extended version.}
    \label{fig:counter-ex}
\end{figure}

This being said, under an extra condition, we keep the results of Theorem~\ref{thm:cns} for the extended version.

\begin{theorem}\label{thm:extension_cns}
    Consider an instance $\calI$ with an arbitrary set $\calS$ of non-overlapping schedules, verifying $|U(\calI)|=q$, i.e., there are $q$ indices $i$ such that $0 \in [s_i,e_i)$. 
    The following statements are equivalent:
    \begin{enumerate}[label=\textup{(\alph*')}]
        \item There exists a feasible assignment with a worker performing each task at least once. \label{ext_cond1}
        \item There exists a balanced feasible assignment. \label{ext_cond2}
        \item There exists a periodic balanced feasible assignment with period bounded by $q^2q!$. \label{ext_cond3}
    \end{enumerate}
\end{theorem}

Note that the implication $\ref{ext_cond2}\Rightarrow\ref{ext_cond3}$ is a direct consequence of Theorem~\ref{thm:periodic} and the implication $\ref{ext_cond3}\Rightarrow\ref{ext_cond1}$ is straightforward, as for Theorem~\ref{thm:cns}. The proof focuses on showing $\ref{ext_cond1}\Rightarrow\ref{ext_cond2}$.

\subsection*{Acknowledgments} This research was partially supported by the SNCF as part of a CIFRE PhD. The authors are grateful to Rolf Nelson van Lieshout for suggesting the possible existence of periodic solutions of period $q$ in the context of Theorem~\ref{thm:cns}, and to Jean-François Delmas for pointing out the situation of ``normal numbers,'' whose existence status shares some similarity with that of balanced feasible assignments (see Section~\ref{subsec:all}).

\section{Main tools}
This section introduces tools that will be useful for the proofs of the theorems.
The results and the constructions of this section are valid for both the basic and extended versions, except for Lemma~\ref{lem:connected_unconstrained} which only concerns the basic version.

\subsection{Building a new feasible assignment from a sequence of feasible assignments}\label{sec:feasible_permutation}
Let $\calI$ be an instance such that $|U(\calI)|=q$.
For every feasible assignment $f$ and every $r \in \Z_{>0}$, we introduce the map $\varphi_{f,r} \colon i \in U(\calI) \mapsto f(i,r) \in [q]$. The assumption $|U(\calI)| = q$ makes this map $\varphi_{f,r}$ a bijection.
In the proofs, we will build new feasible assignments from sequences of feasible assignments. Let $f_1,f_2,\ldots$ be an infinite sequence of feasible assignments. Define inductively permutations $\pi_r$ of $[q]$ by the equation $\pi_{r+1} = \pi_r\circ\varphi_{f_r,2}\circ \varphi_{f_{r+1}, 1}^{-1}$, where $\pi_1$ is an arbitrary permutation of $[q]$. 
This implies in particular 
\begin{equation}\label{eq:pi}
    (\pi_{r+1}\circ f_{r+1})(i,1) = (\pi_{r}\circ f_{r})(i,2) \quad \forall i\in U(\calI)\, .
\end{equation}

\begin{lemma}\label{lem:seq}
    The map $(i,r) \mapsto (\pi_r \circ f_r)(i,1)$ is a feasible assignment.
\end{lemma}

\begin{proof}
    
    %Since $f_r$ is a feasible assignment, then all the schedules it induces belong to $\calS$. Thus, if $g\colon(i,r)\mapsto (\pi_r\circ f_r)(i,1)$ satisfies the non-overlapping property, it also induces schedules belonging to $\calS$ (since it induces the same schedules as $f_r$).    
    Let us show that $g\colon(i,r)\mapsto (\pi_r\circ f_r)(i,1)$ satisfies the non-overlapping property by checking the contrapositive of~\eqref{eq:feasible}.
    Consider $i,i'\in [n]$ with $i\neq i'$ and $r,r'\in\Z_{>0}$. Suppose $g(i,r) = g(i',r')$, i.e., $(\pi_r\circ f_r)(i,1) = (\pi_{r'}\circ f_{r'})(i',1)$. Without loss of generality, suppose that $r \leq r'$. 
    
    Consider first the case when $r=r'$. 
    Since $\pi_r$ is a permutation, we have $f_r(i,1) = f_r(i',1)$. Since $f_r$ is feasible, then the contrapositive holds for $f_r$, namely, $[s_i+1,e_i+1) \cap [s_{i'}+1,e_{i'}+1) = \varnothing$, which is equivalent to $[s_i+r,e_i+r) \cap [s_{i'}+r',e_{i'}+r') = \varnothing$, as desired.

    Consider now the case when $r+1 = r'$. Note first that if $i'\notin U(\calI)$, then $e_i +r < s_{i'}+r+1$ and so $[s_i+r,e_i+r) \cap [s_{i'}+r',e_{i'}+r') = \varnothing$.
    Suppose now that $i'\in U(\calI)$. Using the definition of $\pi_{r+1}$, we have $(\pi_r\circ f_r)(i,1) = (\pi_r\circ f_r)(i',2)$. Since $\pi_r$ is a permutation, then $f_r(i,1)=f_r(i',2)$. The contrapositive holds for $f_r$ feasible assignment, so $[s_i+r,e_i+r) \cap [s_{i'}+r',e_{i'}+r') = \varnothing$.

    Finally, if $r + 2 \leq r'$, then $[s_i+r,e_i+r) \cap [s_{i'}+r',e_{i'}+r')$ is necessarily empty. Therefore, $g$ is non-overlapping, and thus feasible for the basic version.
    For the extended version, it remains to check that the schedules induced by $g$ belong to the set $\calS$ of schedules given in input.
    
    Let $(T,W)$ be a schedule induced by $g$ for the worker $j$ on week $r$. 
    %First, for all tasks $i\in [n]$, we have $g(i,r)=j$ if and only if $f_r(i,1) = \pi_r^{-1}(j)$. Then, 
    If $W=\varnothing$, then the pair $(T,W)$ is exactly the schedule induced by $f_r$ for worker $\pi_r^{-1}(j)$ on the first week. Otherwise, let $i$ be the task in $U(\calI)$ such that $W=\{i\}$. By~\eqref{eq:pi}, we have $g(i,r+1) = (\pi_r \circ f_r)(i,2)$. Hence, $(T, W)$ is the schedule induced by $f_r$ for worker $\pi_r^{-1}(j)$ on the first week. 
    Therefore, every schedule induced by $g$ is also induced by some $f_r$, which means that this schedule belongs to $\calS$.
\end{proof}

%\subsection{Encoding assignments as a directed graph: definition and fundamental property of $D^{\calI, \calF}$}\label{sec:digraph}
\subsection{Encoding assignments as a directed graph: definition and properties}\label{sec:digraph}
Let $\calI$ be an instance such that $|U(\calI)|=q$ and $\calF$ be an arbitrary set of feasible assignments for $\calI$. 
We introduce a directed multi-graph $D^{\calI, \calF}$ which will be useful in the rest of the paper. Its vertex set is $U(\calI)$. Its arc set $A^\calF$ is obtained by introducing $q$ arcs for each $f \in \calF$: an arc from $i \in U(\calI)$ to $i' \in U(\calI)$ whenever $f(i,1) = f(i', 2)$---repetitions are allowed and give rise to parallel arcs---; such an arc is labeled with $f$. 

Note the following properties:
\begin{itemize}
    \item The $q$ arcs labeled with a same feasible assignment $f$ form a collection of vertex-disjoint directed cycles (possibly loops): each vertex is by construction the head of exactly one arc labeled with $f$ and the tail of exactly one arc labeled with $f$.
    \item The number of arcs in $A^\calF$ is $q|\calF|$.
    \item When $D^{\calI, \calF}$ is weakly connected---i.e., the underlying undirected graph is connected---it is also strongly connected and Eulerian.
    %\item Suppose $\calF' \subseteq \calF$. If $D^{\calI, \calF'}$ is Eulerian, then $D^{\calI, \calF}$ is Eulerian as well.
\end{itemize}
We finish by another property of this graph, which requires a proper proof. This property, stated as a lemma, will actually be used within the proofs of Lemmas~\ref{lem:connected_unconstrained} and \ref{lem:connected_balanced}, in the next section. It deals with the notion of ``universal'' sets of feasible assignments we introduce now. A set $\calF$ of feasible assignments is {\em universal} if the existence of a feasible assignment $f$ such that $f(i,1)=f(i',2)$ for some $i,i'\in U(\calI)$ implies the existence of a feasible assignment $g$ in $\calF$ such that $g(i,1)=g(i',2)$. In other words, given two sets $\calF,\calF'$ of feasible assignments, with $\calF$ being universal, if there is an arc in $A^{\calF'}$ from a vertex $i$ to a vertex $i'$, then there is such an arc in $A^\calF$ as well. Note that universal sets of feasible assignments always exist, and can be built with a straightforward greedy procedure.

%For each $i, i'\in  U(\calI)$, consider all feasible assignments $f$ such that $f(i,1)=f(i',2)$, and pick a representative of such assignments, if any, and put it in $\calF_\calI$ if it does not already belong to $\calF_\calI$.

\begin{lemma}\label{lem:induced_walk}
    Let $\calF$ be a set of feasible assignments. Consider the sequence of tasks performed by a worker in some feasible assignment (not necessarily in $\calF$). If $\calF$ is universal, then the induced sequence in $U(\calI)$ translates into a walk of $D^{\calI, \calF}$.
\end{lemma}

\begin{proof}
    Let $f$ be a feasible assignment for $\calI$ and $j$ be a worker. Suppose that $\calF$ is universal. Consider the sequence of tasks performed by $j$ in $f$ and write $i_1,i_2,\ldots$ the induced sequence in $U(\calI)$. 
    This means that, for every $r\in \Z_{>0}$, we have $f(i_r,r)=j=f(i_{r+1},r+1)$. Since $\calF$ is universal, there exists an arc from $i_r$ to $i_{r+1}$ in $D^{\calI,\calF}$ (the assignment $g_h\colon(i,r)\mapsto f(i,r+h-1)$ is feasible and verifies $g_h(i_h,1)=g_{h}(i_{h+1},2)$). Therefore, the sequence $i_1,i_2,\ldots$ translates into a walk in $D^{\calI, \calF}$.
\end{proof}

\subsection{From a general instance \texorpdfstring{$\calI$}{I} to one \texorpdfstring{$\calI'$}{I'} satisfying \texorpdfstring{$|U(\calI')|=q$}{|U(I')|=q}}\label{sec:transfo}
In this section, we describe a transformation of an instance $\calI$ into an instance $\calI'$, with the same number of workers $q$, satisfying $|U(\calI')|=q$. For some cases, this is a desirable property as already hinted in the previous sections. From now on we fix an instance $\calI$. Set 
\[
 n' \coloneqq n+q-|U(\calI)|\quad \text{and}\quad 
 [s_i,e_i) \coloneqq [0,\varepsilon) \text{ for $i \in \{n+1,\ldots,n'\}$, }
 \]
where $\varepsilon>0$ is chosen small enough, i.e., $\varepsilon \leq  \min\bigl(\{s_i \colon i\in [n]\setminus U(\calI)\}\cup\{s_i+1 \colon i\in U(\calI)\}\bigl)$.
%$\varepsilon \coloneqq  \min\bigl(\{s_i \colon i\in [n], s_i >0\}\cup\{s_i+1 \colon i\in [n], s_i \leq 0\}\bigl)$. 
We can interpret the intervals with $i \in \{n+1,\ldots,n'\}$ as new tasks, called {\em fictitious}. For the extended version, we can also get a new set $\calS'$ of valid schedules in a natural way by adding in all possible ways these new tasks to the original valid schedules without creating any overlap. This gives rise to a new instance $\calI'=\left(\bigl([s_i,e_i)\bigl)_{i\in [n']}, q\right)$, or $\calI'=\left(\bigl([s_i,e_i)\bigl)_{i\in [n']}, q, \calS'\right)$ for the extended version. (Note that $\calS'$ can actually be defined in a formal way: for each pair $(T,W) \in \calS$, add to $\calS'$ the pair $(T',W')$ with
\begin{itemize}
    \item $T'=T$ when $T \cap U(\calI) \neq \varnothing$.
    \item $T'=T\cup\{i\}$ for all $i \in \{n+1,\ldots,n'\}$, when $T \cap U(\calI) = \varnothing$.
    \item $W' = W$ when $W\neq \varnothing$.
    \item $W' = \{i\}$ for all $i \in \{n+1,\ldots,n'\}$, when $W=\varnothing$.)
\end{itemize}
An illustration of this transformation is given in Figure~\ref{fig:transfo}.
\begin{figure}
    \centering
    \begin{subfigure}{\textwidth}
    \centering
    \includegraphics[width = \textwidth]{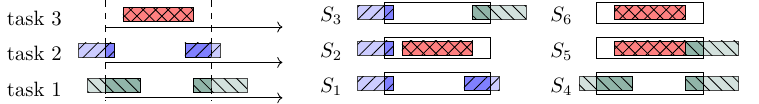}
    \caption{An instance $\calI$ with three tasks ($n=3$), a set of schedules $\calS = \{S_k\colon k\in[6]\}$, and four workers ($q=4$).}
    \end{subfigure}

    \vspace{5mm}
    \begin{subfigure}{\textwidth}
    \centering
        \includegraphics[width = \textwidth]{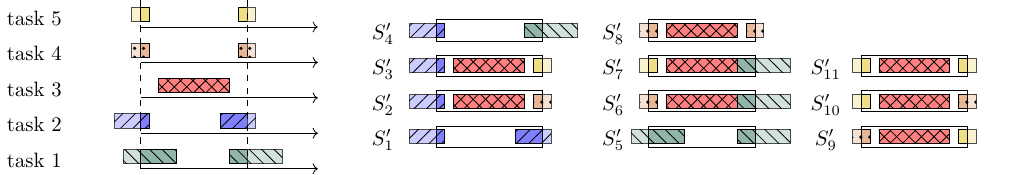}
    \caption{The new instance $\calI'$ with five tasks $(n'=5)$, a set of schedules $\calS'=\{S'_k\colon k\in[11]\}$, and four workers ($q=4$).}
    \end{subfigure}

    \caption{An illustration of the transformation of an instance $\calI$ into a new instance $\calI'$ verifying $|U(\calI')|=q$.}
    \label{fig:transfo}
\end{figure}

The next two lemmas establish relations between the assignments for $\calI$ and for $\calI'$ especially regarding their feasibility and balancedness.

\begin{lemma}\label{lem:transfo_feasible}
    If there exists a feasible assignment $f$ for $\calI$, then there exists a feasible assignment for $\calI'$ whose restriction to the tasks in $\calI$ is $f$. Conversely, if there exists a feasible assignment $f'$ for $\calI'$, then its restriction to the tasks in $\calI$ is a feasible assignment. 
\end{lemma}

\begin{proof}
    Let $f$ be a feasible assignment for $\calI$. The assignment $f$ can be extended on the fictitious tasks as follows: for each positive integer $r$, the assignment $f$ leaves a number of workers idle at time $r$ equal to the number of fictitious tasks; extending $f$ arbitrarily by assigning the $r$th occurrence of these tasks to these workers leads to a feasible assignment for $\calI'$. (The number $\varepsilon$ and the set $\calS'$, in the extended version, have been chosen so that this does not create any conflict.)
    
    Conversely, given a feasible assignment $f'$ for $\calI'$, its restriction to the tasks of $\calI$ is obviously feasible. (In the extended version, it comes from the definition of $\calS'$.)    
\end{proof}

The next lemma is immediate to prove and adds relevance to the transformation.
\begin{lemma}\label{lem:transfo_balanced}
    If there exists a balanced feasible assignment $f'$ for $\calI'$, then there exists a balanced feasible assignment $f$ for $\calI$. Moreover, if $f'$ is periodic, then $f$ can be chosen periodic with the same period as $f'$.
\end{lemma}
\begin{proof}
    If there exists a balanced feasible assignment $f'$ for $\calI'$, then the restriction of this assignment to the tasks of $\calI$ is feasible, balanced, and periodic with same period as $f'$ when $f'$ is periodic.
\end{proof}

The next lemma provides us with a connectivity result for the basic version of the problem. The notion of ``universal'' set has been introduced in Section~\ref{sec:digraph}.

\begin{lemma}\label{lem:connected_unconstrained}
    Assume that $\calI$ is without schedules. Let $\calF$ be a universal set of feasible assignments for $\calI'$. If there exists a feasible assignment for $\calI$ with a worker performing each task at least once, then $D^{\calI',\calF}$ is Eulerian.
\end{lemma}

\begin{proof}
Suppose there exists a feasible assignment for $\calI$ with a worker performing each task at least once. With Lemma~\ref{lem:transfo_feasible}, there exists a feasible assignment $f'$ for $\calI'$ with a worker $j^\star$ performing each task of $\calI$ at least once.
We are going to show that $D^{\calI',\calF}$ is weakly connected. This will be enough with the remark in Section~\ref{sec:digraph}. According to Lemma~\ref{lem:induced_walk} applied to $f'$ and worker $j^\star$, there is a walk in $D^{\calI',\calF}$ visiting every vertex of $U(\calI)$. The remainder of the proof is devoted to showing that every vertex in $U(\calI')\setminus U(\calI)$ has a neighbor in $U(\calI)$.
Let $i'$ be a task in $U(\calI')\setminus U(\calI)$.

Suppose first that the worker $j^\star$ performs no fictitious task.
Let $i$ be a task of $U(\calI)$ with minimum ending time $e_i$ and $r$ be an integer such that $f'(i,r)=j^\star$. Consider the assignment $f''$ defined by $f''(\cdot,h)\coloneqq f'(\cdot,r+h)$ for all $h\in\Z_{>0}$. (Note that it is just the assignment obtained by ``forgetting'' the first $r$ weeks.) Then swap the first occurrences of $i$ and $i'$ in $f''$ to obtain yet another assignment. This new assignment is feasible:
the task $i'$ overlaps no other task by the choice of $\varepsilon$; the task $i$ overlaps no other task either, otherwise this would prevent $j^\star$ to perform all tasks in $f$ while performing no fictitious task (by the choice of $i$).
This shows the existence of an arc from $i'$ to $U(\calI)$ in $D^{\calI',\calF}$, since $r$ was chosen so that $f'(i,r)=j^\star$.

Suppose second that the worker $j^\star$ performs at least one fictitious task $i''$ and let $r$ be such that $f'(i'',r)=j^\star$. Since swapping the $r$th occurrences of $i'$ and $i''$ in $f'$ results in a feasible assignment, there exists an arc from $i'$ to $U(\calI)$ or the other way around.
\end{proof}

The next lemma provides us with another connectivity result, valid for the extended version.
\begin{lemma}\label{lem:connected_balanced}
    Let $\calF$ be a universal set of feasible assignments for $\calI'$.
    If there exists a balanced feasible assignment for $\calI$, then $D^{\calI',\calF}$ is Eulerian.
\end{lemma}
\begin{proof}
    Suppose there exists a balanced feasible assignment for $\calI$. In such an assignment, each worker performs each task of $\calI$ at least once. With Lemma~\ref{lem:transfo_feasible}, there exists a feasible assignment $f'$ for $\calI'$ with all workers performing each task of $\calI$ at least once. Since the first occurrence of every fictitious task is assigned to a worker who performs every task in $U(\calI)$ later on in $f'$, every fictitious task is the first vertex of a walk visiting every vertex in $U(\calI)$ (using Lemma~\ref{lem:induced_walk}). In case there is no fictitious task, there is a walk visiting every vertex in $U(\calI)$ as well, just by considering any worker. The graph $D^{\calI',\calF}$ is thus in any case weakly connected, and we conclude with the remark in Section~\ref{sec:digraph}.
\end{proof}

\section{Proof of Theorem~\texorpdfstring{\ref{thm:cns}}{1.1}}\label{sec:thm1}
In this section, we prove Theorem~\ref{thm:cns} for an instance $\calI=\left(\bigl([s_i,e_i)\bigl)_{i\in[n]}, q\right)$ without schedules.
As discussed in Section~\ref{sec:intro}, we are only left with proving the implication $\ref{cond1}\Rightarrow\ref{cond3}$.

For an instance $\calI$, define $\overline\calF_\calI$ as follows. 
Some feasible assignments induce the same assignment on the first week. For each assignment on the first week, pick a representative among the set of feasible assignments and put it in $\overline\calF_{\calI}$. Remark that $\overline{\calF_\calI}$ is universal.

Consider an instance $\calI$ such that $|U(\calI)|=q$ and a feasible assignment $f$. Suppose we are given a ``time'' $t \in (1,2]$ and two tasks $i_1,i_2 \in U(\calI)$, belonging to distinct cycles in $D^{\calI, \overline{\calF}_{\calI}}$ labeled with $f$, such that $f(i_1,1)$ and $f(i_2,1)$ are idle at time $t$. We define the {\em merge operation at time $t$ between the tasks $i_1$ and $i_2$}, providing a new feasible assignment $\tilde f$ by swapping all tasks assigned to the two workers after time $t$ as follows:

$$
\tilde{f}(i,r) \coloneqq \begin{cases}
    f(i,r) &\text{ if } f(i,r)\notin \{f(i_1,1), f(i_2,1)\},\\
    f(i,r) &\text{ if } e_i+r\leq t,\\
    f(i_1,1) &\text{ if }e_i+r> t \text{ and }f(i,r)=f(i_2,1),\\
    f(i_2,1) &\text{ if }e_i+r> t \text{ and }f(i,r)=f(i_1,1).\\
\end{cases}
$$
Note that no task $i$ performed by $f(i_1,1)$ or $f(i_2,1)$ is such that $t\in [s_{i}+1, e_{i}+1)$ since both workers are idle at time $t$. This makes the assignment $\tilde f$ feasible.

An illustration of a merge operation is given in Figure~\ref{fig:merge}.

\begin{figure}
    \begin{subfigure}{\textwidth}
        \includegraphics[width = \textwidth]{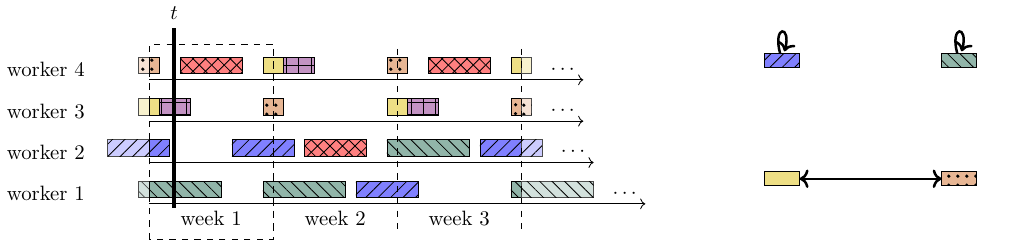}
    \caption{A feasible assignment $f\in\overline{\calF}_\calI$ for an instance $\calI$ such that $|U(\calI)|=q$. The arcs labeled by $f$ in $D^{\calI, \overline{\calF}_\calI}$ are illustrated on the right. A merge operation can be performed at time $t$ between the two tasks in hatched blue (worker $2$) and in dotted orange (worker $4$).}
    \label{fig:merge_1}
    \end{subfigure}
    
    \begin{subfigure}{\textwidth}
    \includegraphics[width = \textwidth]{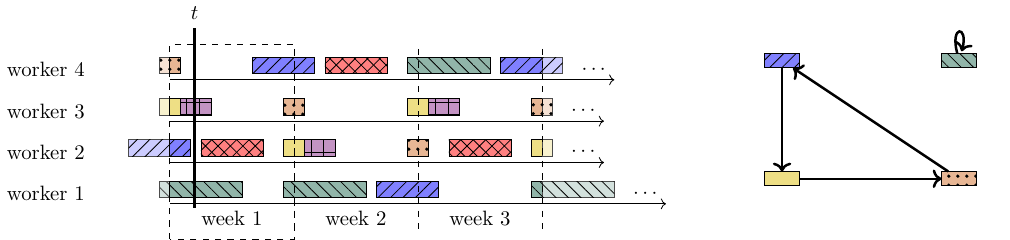}
    \caption{The assignment $\tilde f$ resulting from the merge operation of Figure~\ref{fig:merge_1}. The arcs labeled by the representative of $\tilde f$ in $D^{\calI, \overline{\calF}_\calI}$ are illustrated on the right.}
    \end{subfigure}
    \caption{An illustration of a merge operation on a feasible assignment.}
    \label{fig:merge}
\end{figure}
%Denote by $\tilde{f'}$ the assignment in $\overline{\calF_{\calI}}$ representative of $\tilde f$ in $\overline{\calF_{\calI}}$.

 \begin{lemma}\label{lem:variant}
    Let $\calI$ be an instance such that $|U(\calI)|=q$. Let $f$ be a feasible assignment and $\tilde f$ be the feasible assignment resulting from a merge operation on $f$.
    The number of cycles in $D^{\calI, \overline{\calF}_{\calI}}$ labeled with the representative of $\tilde f$ in $\overline{\calF}_{\calI}$ is smaller by one unit than the number of cycles labeled with the representative of $f$ in $\overline{\calF}_{\calI}$.
\end{lemma}

\begin{proof}
    Suppose that the merge operation performed to obtain $\tilde f$ from $f$ is done at time $t$ between the tasks $i_1$ and $i_2$.
    Consider the arc of tail $i_1$ (resp.\ $i_2$) labeled with the representative of $f$ in $D^{\calI, \overline{\calF}_{\calI}}$ and denote by $i_1'$ (resp.\ $i_2'$) the head of this arc. We have $\tilde f(i_1,1)=\tilde f(i_2',2)$ using the second and third cases of the definition of $\tilde f$ . Thus the arc leaving $i_1$ labeled with the representative of $\tilde f$ in $\overline{\calF}_{\calI}$ has head $i'_2$. Similarly, with the second and fourth cases, we have $\tilde f(i_2,1)=\tilde f(i_1',2)$. Thus the arc leaving $i_2$ labeled with the representative of $\tilde f$ in $\overline{\calF}_{\calI}$ has head $i'_1$. Consider the arcs different from the two going from $i_1$ to $i_2'$ and from $i_2$ to $i_1'$ and labeled with the representative of $\tilde f$ in $\overline{\calF}_{\calI}$. Each of these arcs is parallel to an arc labeled with $f$, according to the first case of the definition of $\tilde f$. Therefore the number of cycles formed by the arcs labeled with the representative of $\tilde f$ in $\overline{\calF}_{\calI}$ is smaller by one unit than for the representative of $f$. 
\end{proof}

For the next two proofs (and a few other times), the following notation will be useful. For a feasible assignment $f$ and an integer $r$, define $u_{f,r}(i)$ for $i\in[n]$ as the unique task in $U(\calI)$ whose $r$th occurrence is assigned to the same worker as the $r$th occurrence of $i$. Formally, $u_{f,r}(i)$ is the unique task in $U(\calI)$ such that $f(u_{f,r}(i),r) = f(i,r)$.

The following two lemmas will be useful to establish Theorems~\ref{thm:cns} and~\ref{thm:algo}.
\begin{lemma}\label{lem:big-cycle}
    Let $\calI$ be an instance such that $|U(\calI)|=q$.
    Suppose that $D^{\calI, \overline{\calF}_{\calI}}$ is Eulerian.
    If a feasible assignment cannot be subject to a merge operation, then the arcs labeled with its representative in $\overline{\calF}_{\calI}$ form a Hamiltonian cycle.
\end{lemma}

%iA natural approach to the proof of this lemma would be to show that if there is an assignment $f^\star$ labeling to distinct cycles, then there are two workers idle at the same time in $f^\star$ with which a merge operation could be performed. In essence, {\blue this could be obtained by adapting the following proof, assuming that $f^\star$ labels a minimal number of cycles. To our knowledge, such an approach does not make the proof more straightforward.} \todo{complete}
In essence, the proof of this lemma assumes that there are more than two cycles and identifies a common idle time that can be used to perform a merge operation. For some technical reasons, we were however not able to make this scheme as explicit as we wished in the proof.

\begin{proof}[Proof of Lemma~\ref{lem:big-cycle}]
    Take a feasible assignment that cannot be subject to a merge operation and denote by $f^\star$ its representative in $\overline{\calF}_{\calI}$. For a feasible assignment $g$, denote by $\rho_g(i)$ the task in $U(\calI)$ whose first occurrence is assigned to the worker who starts the task $i$ on the first week. Formally,
    \[
    \rho_g(i) = \begin{cases}
        (\varphi_{g,1}^{-1}\circ\varphi_{g,2})(i)&\text{ if }i\in U(\calI),\\
        u_{g,1}(i)&\text{ if }i\notin U(\calI).
    \end{cases}
    \]
    We show now the following property: {\em For every feasible assignment $g$ and task $i\in [n]$, there is a cycle labeled with $f^\star$ containing the vertices $u_{f^\star,1}(i)$ and $\rho_g(i)$ simultaneously.} 
    Once this property is shown, we finish the proof as follows. Suppose by contradiction that there are at least two cycles labeled with $f^\star$. Choose two such cycles with an arc $a$ of $D^{\calI, \overline \calF_\calI}$ connecting them. Such cycles and arc exist because $D^{\calI, \overline \calF_\calI}$ is Eulerian. Let $g$ be the assignment labeling $a$, and $i'$ and $i$ be respectively the tail and the head of $a$. Using the property above, there is a cycle labeled with $f^\star$ containing the vertices $u_{f^\star,1}(i) = i$ and $\rho_g(i) = i'$ simultaneously. This contradicts the fact that $i$ and $i'$ belong to different cycles labeled with $f^\star$. 

    We will proceed now by induction to prove the property. To do so, for each task $i\in [n]$, we define 
    \[
    \tilde s_i \coloneqq \begin{cases}
        s_i+1 &\text{ if } i\in U(\calI),\\
        s_i &\text{ if } i\notin U(\calI),\\
    \end{cases}
    \quad\text{ and }\quad
    \tilde e_i \coloneqq \begin{cases}
        e_i+1 &\text{ if } i\in U(\calI),\\
        e_i &\text{ if } i\notin U(\calI).\\
    \end{cases}
    \]
    We sort the tasks $i_1,\ldots, i_n$ by ascending order of $\tilde s_i$. Let $g$ be a feasible assignment, and let us show by induction on $\ell$ that there is a cycle labeled with $f^\star$ containing the vertices $u_{f^\star,1}(i_\ell)$ and $\rho_g(i_\ell)$ simultaneously.
    
    First, let us prove the property for $\ell=1$.
    By the definition of $f^\star$, every task $i'\in U(\calI)$ verifying $[s_{i'},e_{i'})\cap[\tilde s_{i_1},\tilde e_{i_1}) = \varnothing$ is in the same cycle as $u_{f^\star,1}(i_1)$ (otherwise, a merge operation could be performed right before time $\tilde s_{i_1}+1$ between the tasks $u_{f^\star,1}(i_1)$ and $i'$). Since $g$ is a feasible assignment, we have $[s_{\rho_g(i_1)},e_{\rho_g(i_1)})\cap[\tilde s_{i_1},\tilde e_{i_1}) = \varnothing$. Therefore there is a cycle labeled with $f^\star$ containing the vertices $u_{f^\star,1}(i_1)$ and $\rho_g(i_1)$ simultaneously.
    
    Suppose now that the property holds up to $\ell\geq 1$. Let us prove the property for $\ell+1$.
    Write $C_{\ell+1}$ the cycle labeled with $f^\star$ containing the vertex $u_{f^\star,1}(i_{\ell+1})$. 
    For $i\in U(\calI)\setminus C_{\ell+1}$ define $\chi(i)$ as follows. In the assignment $f^\star$, the worker $f^\star(i,1)$ is not idle at time $\tilde s_{i_{\ell+1}} +1$ (otherwise, a merge operation could be performed right before time $\tilde s_{i_{\ell+1}}+1$ between the tasks $u_{f^\star,1}(i_{\ell+1})$ and $i$). Thus, the worker $f^\star(i,1)$ is busy, performing either:
    \begin{enumerate}
    \item \label{item:chi_1}the first occurrence of the task $i \in U(\calI)$. 
    \item \label{item:chi_2}the first occurrence of a task $i_k\notin U(\calI)$. Note that $k\leq \ell$ in that case.
    \item \label{item:chi_3}the second occurrence of a task $i_k\in U(\calI)$. Note that $k\leq \ell$ in that case.
    \end{enumerate}
    Set
    \[
    \chi(i)\coloneqq \begin{cases}
        i&\text{in case~\eqref{item:chi_1},}\\
        \rho_g(i_k)&\text{in cases~\eqref{item:chi_2} and ~\eqref{item:chi_3}.}
    \end{cases}
    \]
    An illustration is given in Figure~\ref{fig:chi}. Notice that $g(\chi(i),1)$ performs a task different from $i_{\ell+1}$ at time $\tilde s_{\ell+1}$ in $g$.

    The map $\chi$ defines a self-map of $U(\calI)\setminus C_{\ell+1}$: in case~\eqref{item:chi_1}, the task $\chi(i) = i$ does not belong to $C_{\ell+1}$ and in cases~\eqref{item:chi_2} and \eqref{item:chi_3}, by induction, the task $\chi(i) = \rho(i_k)$ belongs to the same cycle as $u_{f^\star,1}(i_k) = i$, which is different than $C_{\ell+1}$ by definition of $i$. We claim that the map $\chi$ is actually injective. Let $i\neq i'\in U(\calI)\setminus C_{\ell+1}$. As noted above, with tasks chosen this way, the workers $f^\star(i,1)$ and $f^\star(i',1)$ are busy at time $\tilde s_{i_{\ell+1}}+1$ in the assignment $f^\star$. Let us check that $\chi(i)\neq \chi(i')$. 
    
    Suppose first that $\chi(i) = i$ and $\chi(i')=i'$ (case~\eqref{item:chi_1} for $i$ and $i'$). Since $i\neq i'$, we have $\chi(i)\neq \chi(i')$.
    Suppose now that $\chi(i) = \rho_g(i_k)$ (case~\eqref{item:chi_2} for $i$) and $\chi(i')=i'$ (case~\eqref{item:chi_1} for $i'$). The tasks $i_k$ and $i'$ overlap because $f^\star(i,1)$ and $f^\star(i',1)$ are busy at time $\tilde s_{i_{\ell+1}}+1$ in $f^\star$ with the first occurrences of respectively $i_k$ and $i'$. Thus $\rho_g(i_k)\neq i'$, i.e., $\chi(i)\neq \chi(i')$.
    Suppose then that $\chi(i) = \rho_g(i_k)$ (case~\eqref{item:chi_2} for $i$) and $\chi(i')=\rho_g(i_{k'})$ (case~\eqref{item:chi_2} for $i'$). The tasks $i_k$ and $i_{k'}$ overlap because $f^\star(i,1)$ and $f^\star(i',1)$ are busy at time $\tilde s_{i_{\ell+1}}+1$ in $f^\star$ with the first occurrences of respectively $i_k$ and $i_{k'}$. The images of overlapping tasks by $\rho_g$ are distinct, and thus $\chi(i)\neq \chi(i')$.
    When $\chi(i)$ or $\chi(i')$ is defined according to case~\eqref{item:chi_3}, the proof of the relation $\chi(i)\neq \chi(i')$ follows along the same lines as for case~\eqref{item:chi_2}.
    
    The map $\chi$ is then bijective, which implies $\{g(\chi(i),1)\colon i\in U(\calI)\setminus C_{\ell+1}\} = \{g(i,1)\colon i\in U(\calI)\setminus C_{\ell+1}\}$. As noted above, the worker $g(\chi(i),1)$ performs a task different from $i_{\ell+1}$ at time $\tilde s_{i_{\ell+1}}+1$ in the assignment $g$.
    Thus, no worker in $\{g(i,1)\colon i\in U(\calI)\setminus C_{\ell+1}\}$ performs a task $i_{\ell+1}$ at time $\tilde s_{i_{\ell+1}}+1$ in the assignment $g$. This implies $\rho_g(i_{\ell+1})\in C_{\ell+1}$, which concludes the proof of the property.
\end{proof}

\begin{figure}
\begin{subfigure}{\textwidth}
    \includegraphics[width=\textwidth]{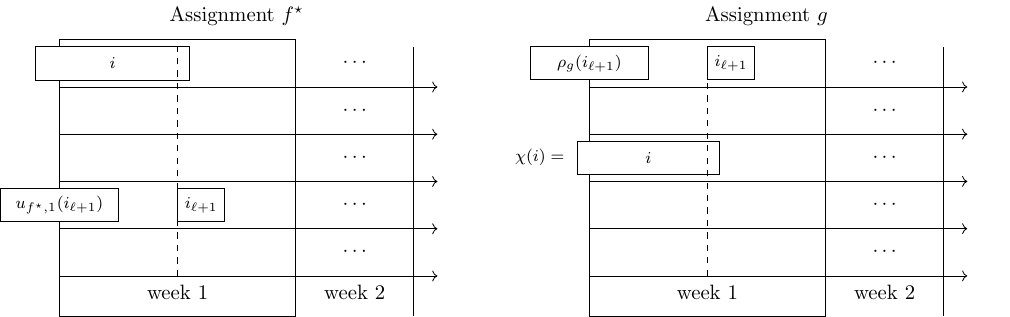}
    \caption{An illustration of the map $\chi$ in the case \eqref{item:chi_1}: $\chi(i)=i$.}
\end{subfigure}
\begin{subfigure}{\textwidth}
    \includegraphics[width=\textwidth]{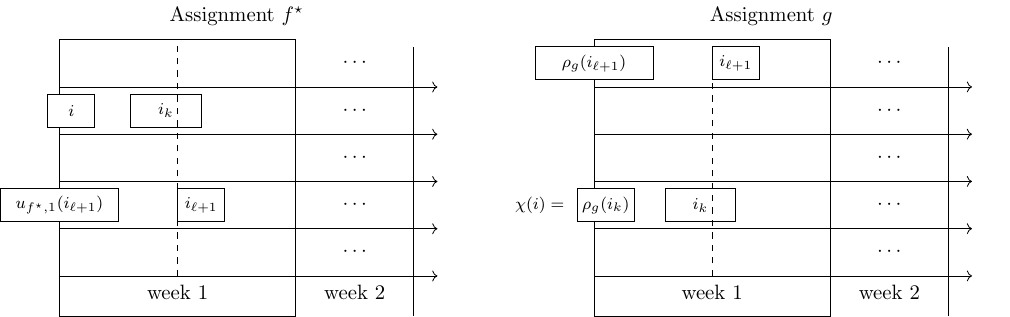}
    \caption{An illustration of the map $\chi$ in the cases~\eqref{item:chi_2} and~\eqref{item:chi_3}: $\chi(i)=\rho_g(i_k)$.}
\end{subfigure}
    \caption{An illustration of the map $\chi$ for the different cases~\eqref{item:chi_1},~\eqref{item:chi_2}, and~\eqref{item:chi_3}.}
    \label{fig:chi}
\end{figure}

\begin{lemma}\label{lem:balanced_q}
    Let $\calI$ be an instance such that $|U(\calI)|=q$ and $\calF$ be a set of feasible assignments for $\calI$.
    If $D^{\calI, \calF}$ is Eulerian, then there exists a periodic balanced feasible assignment for $\calI$ with period $q$.
\end{lemma}
\begin{proof}
    Suppose that $D^{\calI, \calF}$ is Eulerian. Note that $\calF$ is then non-empty (even if $q=1$) and thus that $\overline \calF_\calI$ is non-empty as well. The graph $D^{\calI, \overline\calF_{\calI}}$ contains a subgraph isomorphic to $D^{\calI, \calF}$ and they have the same vertex set. The graph $D^{\calI, \overline\calF_{\calI}}$ is thus weakly connected, and thus Eulerian as noted in Section~\ref{sec:digraph}.

Pick any feasible assignment in $\overline{\calF}_{\calI}$ and consider the feasible assignment obtained from it by performing merge operations as long as possible. According to Lemma~\ref{lem:variant}, at most $q$ such operations can be performed, and thus the process finishes. Lemma~\ref{lem:big-cycle} implies then that there exists a feasible assignment $g\in \overline{\calF}_{\calI}$ labeling arcs of a Hamiltonian cycle $i_1, \ldots, i_q$ in $D^{\calI, \overline{\calF}_{\calI}}$. By definition of $D^{\calI, \overline{\calF}_{\calI}}$, we have $g(i_{k},1)=g(i_{k+1},2)$ for all $k\in[q]$, which can be rewritten as $\varphi_{g,1}(i_{k}) = \varphi_{g,2}(i_{k+1})$ (where $i_{k+q} \coloneqq i_k$).

Consider now the permutations $\pi_r$ as defined in Section~\ref{sec:feasible_permutation} for the infinite sequence $g,g,\ldots$ and set $g'(i,r)\coloneqq (\pi_r\circ g)(i,1)$ for every $i\in[n]$ and every $r\in \Z_{>0}$. According to Lemma~\ref{lem:seq}, this assignment $g'$ is feasible. 

We first show that $g'$ is periodic of period $q$. We have, for all $k\in[q]$ and $r\in \Z_{>0}$,
$$\pi_{r}(\varphi_{g,1}(i_{k})) =  \pi_{r}(\varphi_{g,2}(i_{k+1})) = (\pi_{r+1}\circ\varphi_{g,1}\circ\varphi_{g,2}^{-1})(\varphi_{g,2}(i_{k+1})) = \pi_{r+1}(\varphi_{g,1}(i_{k+1})) \, .$$
By a direct induction, we get $\pi_r(\varphi_{g,1}(i_k))=\pi_{r+h}(\varphi_{g,1}(i_{k+h}))$, for all $r,h\in\Z_{>0}$ and $k\in[q]$. Thus,
\begin{equation}\label{equ:induction}
g'(i_k,r)=\pi_r(g(i_k,1)) = \pi_r(\varphi_{g,1}(i_k))=\pi_{r+h}(\varphi_{g,1}(i_{k+h})) = \pi_{r+h}(g(i_{k+h},1)) = g'(i_{k+h},r+h)\ .
\end{equation}

In particular, we have $g'(i_k,r)=g'(i_{k+q},r+q)=g'(i_k,r+q)$ for every $k\in[q]$ and $r\in \Z_{>0}$.
Therefore, for all $i\in [n]$ and $r\in \Z_{>0}$, we have
\begin{align*}
g'(i,r) &= \pi_r(g(i,1)) = \pi_r(g(u_{g,1}(i),1)) =g'(u_{g,1}(i),r) = g'(u_{g,1}(i),r+q) \\&= \pi_{r+q}(g(u_{g,1}(i),1)) = \pi_{r+q}(g(i,1)) = g'(i,r+q) \, ,   
\end{align*}
as desired.

We then show that $g'$ is balanced. Consider $i \in [n]$ and $j \in [q]$. %Since $g'$ is periodic, we have $\bigl|\{r \in [t] \colon g'(i,r) = j\}\bigl| \geq \lfloor t/q \rfloor$. 
Define $i_k \coloneqq u_{g,1}(i)$ and $i_\ell \coloneqq \varphi_{g',1}^{-1}(j)$. We have 
\[
g'(i,r) = \pi_r(g(i,1)) = \pi_r(g(u_{g,1}(i),1))= g'(u_{g,1}(i),r)= g'(i_k,r) = g'(i_{k-r+1},1) \, ,
\]
where equation~\eqref{equ:induction} is used to get the last equality. Since $\varphi_{g',1}$ is bijective, $g'(i,r) = j$ if and only $i_{k-r+1} = i_\ell$, i.e., if and only if $r \equiv k - \ell + 1 \pmod q$.
From this, we get
\begin{align*}
|\{ r \in [t] \colon g'(i,r) = j \}| &= \sum_{r=1}^t \un(g'(i,r)=j) = \sum_{\substack{r \in [t] \\ r \equiv k-\ell+1 \hspace{-2mm}\pmod q}} \un(g'(i,r)=j) \\ 
&=  |\{r \in [t] \colon r \equiv k-\ell+1 \hspace{-2mm}\pmod q\}| \, ,
\end{align*}
where $\un(\cdot)$ is the indicator function.
%Let $\ell\in [q]$ be such that $i_{\ell} = u_{g,1}(i)$
%We have $g'(i,r)=g'(u_{g,1}(i),r)=g'(i_\ell,r)$. Using equation~\eqref{equ:induction}, we thus have $g'(i_\ell,r) = g'(i_{\ell-r+1},1) = \varphi_{g',1}(i_{\ell-r+1})$. Since $\varphi_{g',1}$ is bijective, $g'(i,r)=j$ if and only if $i_{\ell-r+1} = \varphi_{g',1}^{-1}(j)$
%Let $\ell\in [q]$ be such that $i_{\ell} = u_{g,1}(i)$. By the bijectivity of $\varphi_{g',1}$, there is a unique index $k\in[q]$ such that $g'(i_k,1) = j$. Using equation~\eqref{equ:induction}, we have $g'(i_{k+r-1},r) = j$ for all $r \in \Z_{>0}$. Since $g'(i,r)=g'(u_{g,1}(i),r)=g'(i_\ell,r)$, this implies that $g'(i,r) = j$ if and only if $k + r-1$ is equal to $\ell$ modulo $q$, i.e., if and only if $r$ is equal to $\ell-k+1$ modulo $q$. 
Thus, $\lceil t/q \rceil \geq \bigl|\{r \in [t] \colon g'(i,r) = j\}\bigl| \geq \lfloor t/q \rfloor$ holds. Passing to the limit shows that $\lim_{t\to\infty} \frac1t\bigl|\{r \in [t] \colon g'(i,r) = j\}\bigl|$ exists and is equal to $\frac1q$. The assignment $g'$ is thus balanced.
\end{proof}

\begin{proof}[Proof of $\ref{cond1}\Rightarrow\ref{cond3}$ in Theorem~\ref{thm:cns}] 
Consider an instance $\calI$ (without schedules), and denote by $\calI'$ the associated instance defined as in Section~\ref{sec:transfo}. Suppose we are given a feasible assignment for $\calI$ with a worker performing each task at least once.
According to Lemma~\ref{lem:connected_unconstrained}, since $\overline \calF_{\calI'}$ is a universal set of feasible assignments as noted at the beginning of this section, the graph $D^{\calI',\overline\calF_{\calI'}}$ is Eulerian. Using Lemma~\ref{lem:balanced_q}, there exists a periodic balanced feasible assignment for $\calI'$ with period $q$. Thus, according to Lemma~\ref{lem:transfo_balanced}, there exists a periodic balanced feasible assignment for $\calI$ with period $q$.
\end{proof}

\section{Proof of Theorem~\texorpdfstring{\ref{thm:algo}}{1.2}}

In this section, we prove Theorem~\ref{thm:algo} for an instance $\calI=\left(\bigl([s_i,e_i)\bigl)_{i\in [n]},q\right)$ without schedules. The proof of this theorem combines the results of Section~\ref{sec:thm1} and the following lemmas. 

\begin{lemma}\label{lem:polynomial_bigq}
    Let $\calI$ be an instance such that $q\geq 2n$. Then there always exists a periodic balanced feasible assignment with period $q$, and a closed formula of such an assignment can be computed in polynomial time.
\end{lemma}

\begin{proof}
    For $i\in[n]$ and $r\in \Z_{>0}$, define \begin{equation}\label{eq:closed_formula}
    f(i,r)\coloneqq \begin{cases}
        r_1+2-2i&\text{ if }0<r_1+2-2i,\\
        r_1+2-2i+q&\text{ if }r_1+2-2i\leq 0,\\
    \end{cases}
    \end{equation}
    where $r_1$ is the remainder of the Euclidean division of $r$ by $q$. This ensures that $r_1 \leq q-1$, and thus that $q-1\geq r_1+2-2i$. So when $0<r_1+2-2i$, we have $r_1+2-2i\in [q]$. Moreover, we have $r_1+2-2i\geq 2-2n\geq 2-q$. So when $r_1+2-2i\leq 0$, we have $r_1+2-2i+q\in [q]$. Therefore in both cases, $f(i,r)\in[q]$ holds, which makes $f$ an assignment. Let us show that $f$ is balanced, feasible, and periodic with period $q$.

    First, we show that $f$ is feasible. Let $i,i'\in [n]$ and $r,r'\in\Z_{>0}$ be such that $f(i,r) = f(i',r')$. Suppose first $|r-r'| > 1$. Then $[s_{i}+r,e_{i}+r) \cap [s_{i'}+r',e_{i'}+r')$ is necessarily empty as noted in Section~\ref{sec:intro}. Suppose now that $|r-r'|\leq 1$. Without loss of generality, suppose that $r \in \{r',r'+1\}$. Let $r_1$ and $r_1'$ be the remainders of the Euclidean divisions of $r$ and $r'$ by $q$. Using~\eqref{eq:closed_formula}, we have then necessarily $r_1 - 2i \equiv r'_1 - 2i' \pmod q$, which can be written $r-2i \equiv r'-2i' \pmod q$. We cannot have $r=r'+1$ because otherwise the equality would rewrite $2i-1 \equiv 2i' \pmod q$, which would imply $2i-1=2i'$, which is impossible. Thus, $r=r'$ and since $2i, 2i'\in[q]$, we have $i=i'$. Therefore, the assignment $f$ is feasible.

    Second, we show that $f$ is periodic with period $q$. Let $i\in[n]$ and $r\in\Z_{>0}$. Writing $r=r_2q+r_1$ the Euclidean division of $r$ by $q$, we have $r+q=(r_2+1)q+r_1$ which is the Euclidean division of $r+q$ by $q$. Therefore $f(i,r)=f(i,r+q)$ which means that $f$ is periodic with period $q$.

    Finally, we show that the assignment $f$ is balanced. Let $i\in [n]$ and $j\in [q]$. We claim that there is a unique $r\in \{q,\ldots,2q-1\}$ such that $f(i,r)=j$. Indeed, we have $r = r_1+q$ where $$r_1 = \begin{cases}
        j+2i-2 &\text{ if } j+2i-2 \in \{0,\ldots, q-1\},\\
        j+2i-2-q &\text{ if } j+2i-2 \in \{q,\ldots, 2q-1\}.\\
    \end{cases}$$ 
    Since $f$ is periodic with period $q$, this shows that $f$ is balanced.

    Furthermore, since the Euclidean division can be performed in polynomial time, the closed formula~\eqref{eq:closed_formula} gives the desired property.
\end{proof}
\begin{figure}
    \centering
    \includegraphics[width=0.4\textwidth]{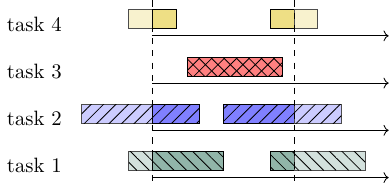}
    
    \vspace{5mm}
    \includegraphics[width=\textwidth]{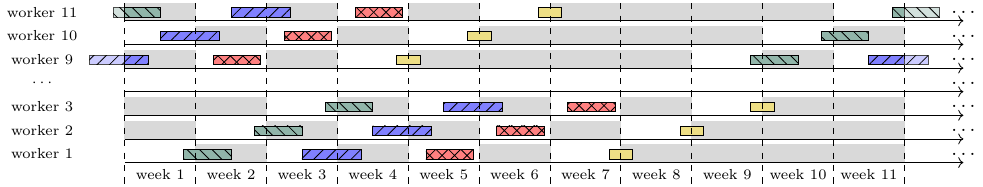}
    \caption{An illustration of the assignment given by equation~\eqref{eq:closed_formula} in the proof of Lemma~\ref{lem:polynomial_bigq}. The instance has $n=4$ tasks (with intervals illustrated above) and $q=11$ workers. The gray areas are weeks on which no tasks are started.}
    \label{fig:enter-label}
\end{figure}

For the next two lemmas, the first week is understood in a broad manner, including the second occurrences of the tasks of $U(\calI)$.

\begin{lemma}\label{lem:interval_coloring}
    Deciding whether there exists a feasible assignment for an instance $\calI$ and computing the first week of such an assignment if it exists can be done in polynomial time.
\end{lemma}

\begin{proof}
    Let $G$ be the interval graph built from all the intervals $[s_i,e_i)$ for $i \in [n]$ together with the intervals $[s_i+1,e_i+1)$ for $i \in U(\calI)$. Computing (if it exists) the first week of a feasible assignment $f$ is equivalent to computing (if it exists) a proper $q$-coloring of $G$, by identifying the colors as the workers. It is well known that this can be done in polynomial time. Note that a proper $q$-coloring of $G$ induces an assignment on the first week which can be extended into a feasible assignment. 
\end{proof}

The next lemma discusses the tractability of the merge operation, defined in Section~\ref{sec:thm1}.

\begin{lemma}\label{lem:merge_polynomial}
    Let $\calI$ be an instance such that $|U(\calI)|=q$ and $f$ be an arbitrary feasible assignment for $\calI$. Given a time $t\in (1,2]$ and two tasks $i_1,i_2\in U(\calI)$, deciding whether a merge operation can be performed on $f$ at time $t$ between $i_1$ and $i_2$ can be done in polynomial time. Moreover, when such an operation can be done, the restriction on the first week of the assignment resulting from such a merge operation can be computed in polynomial time as well.
\end{lemma}

\begin{proof}
    Computing the cycles of in $D^{\calI,\overline\calF_\calI}$ labeled with $\bar f$ the representative of $f$ in $\overline\calF_\calI$ can be done in polynomial time in $O(q)$. Indeed, there is an arc from $i$ to $i'$ in $D^{\calI,\overline\calF_\calI}$ labeled with $\bar f$ if and only if $f(i,1)=f(i',2)$. 
    Checking whether a merge operation can be performed of $f$ at time $t$ between $i_1$ and $i_2$ is equivalent to checking if $f(i_1,1)$ and $f(i_2,1)$ are both idle at time $t$ and if $i_1$ and $i_2$ belong to distinct cycles. Moreover, it is clear from the definition of a merge operation that the first week of the resulting assignment can be computed in polynomial time.
    %Indeed, the graph $D^{\calI,\overline\calF_\calI}$ does not even need to be built to compute these cycles: we only need to find the unique $i'\in U(\calI)$ such that $f(i,1)=f(i',2)$ for all $i\in U(\calI)$. 
    %Checking whether a merge operation can be performed on $f$ is equivalent to checking for each $t\in \{e_1+1,\ldots, e_n+1\}$ and each pair of vertices $i_1,i_2$ in different cycles whether $f(i_1,1)$ and $f(i_2,1)$ are simultaneously idle at time $t$. The latter can be done in polynomial time in $n$ and there are at most $nq^2$ such triplets. Furthermore, if such a triplet exists, the construction of the assignment $\tilde f$, resulting from the merge operation on $f$, on the first week can be done in polynomial time.
\end{proof}

\begin{proof}[Proof of Theorem~\ref{thm:algo}]
    Let $\calI=\left(\bigl([s_i,e_i)\bigl)_{i\in [n]},q\right)$ be an instance (without schedules).
    The case $q\geq 2n$ is dealt with using Lemma~\ref{lem:polynomial_bigq}, which provides a closed formula for an assignment satisfying the desired properties. We can thus assume for the rest of the proof that $q\leq 2n-1$. 
    
    Transform $\calI$ into $\calI'=\left(\bigl([s_i,e_i)\bigl)_{i\in[n']}, q\right)$ as described in Section~\ref{sec:transfo}. This can obviously be done polynomial time.
    Using Lemma~\ref{lem:interval_coloring}, deciding the existence of a feasible assignment for $\calI'$, and computing the first week of such an assignment, if it exists, can be done in polynomial time. If no feasible assignment exists for $\calI'$, then using Lemma~\ref{lem:transfo_feasible}, there is no feasible assignment for $\calI$, and we are done. 
    Thus, suppose from now on the existence of a feasible assignment $f$ for $\calI'$. Perform merge operations as long as possible starting from $f$. Notice that these merge operations are only performed at times in $(1,2]$, hence only the first week of $f$ is needed. An example is given in Figure~\ref{fig:algo}.
    Each merge operation is polynomial according to Lemma~\ref{lem:merge_polynomial} (the time $t$ can in fact be chosen in $\{e_1+1,\ldots,e_{n'}+1\}$) and the number of merge operations is upper bounded by $q$ according to Lemma~\ref{lem:variant}. Hence, in polynomial time, we can get the first week of an assignment $f^\star$ which cannot be subject to a merge operation. 
    Consider the arcs in $D^{\calI',\overline\calF_{\calI'}}$ labeled with the representative of $f^\star$ in $\overline \calF_{\calI'}$. (To get these arcs, only the first week of $f^\star$ is needed.)
    
    Suppose first that they form a Hamiltonian cycle. Following the same lines as in the the proof of Lemma~\ref{lem:balanced_q}, this provides us with a periodic balanced feasible assignment for $\calI'$ with period $q$. An illustration is given in Figure~\ref{fig:hamiltonian_balanced}.

    Suppose now that they do not form a Hamiltonian cycle. Using Lemma~\ref{lem:big-cycle}, the graph $D^{\calI',\overline\calF_{\calI'}}$ is not Eulerian. 
    Since $\overline \calF_{\calI'}$ is a universal set of feasible assignments for $\calI'$, using Lemma~\ref{lem:connected_balanced}, there is no balanced feasible assignment for $\calI$.
    %Combining Lemma~\ref{lem:connected_balanced} with the remark in Section~\ref{sec:digraph} concerning the inclusion $\calF_{\calI'}\subseteq\overline \calF_{\calI'}$, we get that there is no balanced feasible assignment for $\calI$.  
\end{proof}
\begin{figure}
    \begin{subfigure}{\textwidth}
    \includegraphics[width = \textwidth]{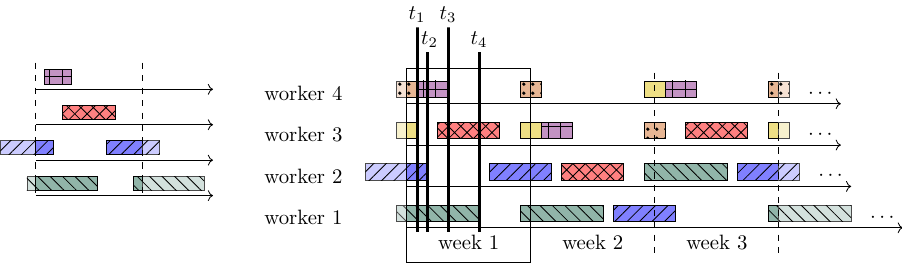}
    \caption{An illustration of an instance $\calI$ (left) with four workers and a feasible assignment $f$ for $\calI'$ (right). A series of merge operations can be performed: at time $t_1$ between the plain yellow (worker $3$) and dotted orange (worker $4$) tasks, then at time $t_2$ between the hatched blue (worker $2$) and plain yellow (worker $4$) tasks, then at time $t_3$ between workers $3$ and $4$, and finally at time $t_4$ between workers $1$ and $4$.}
    \end{subfigure}
    
    \begin{subfigure}{\textwidth}
        \includegraphics[width = \textwidth]{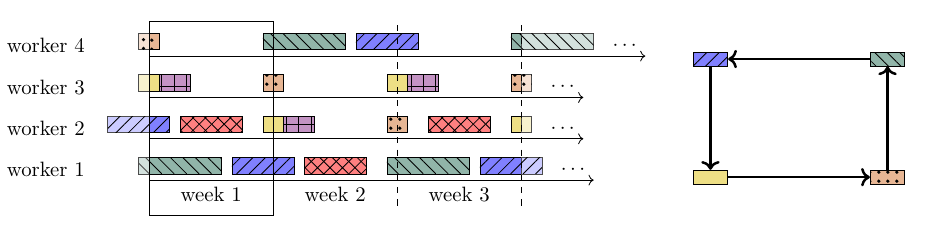}
    \caption{The assignment $f^\star$ resulting from the series of merge operations described above (at time $t_1$, $t_2$, $t_3$, and $t_4$), starting from $f$. The assignment $f^\star$ cannot be subject to a merge operation. The arcs labeled by $f^\star$ in $D^{\calI', \overline{\calF}_{\calI'}}$ form a Hamiltonian cycle.}
    \label{fig:algo_result}
    \end{subfigure}
    \caption{An illustration of a series of merge operations on a feasible assignment, until no merge operation can be performed anymore.}
    \label{fig:algo}
\end{figure}
\begin{figure}
    \centering
    \includegraphics[width = 0.8\textwidth]{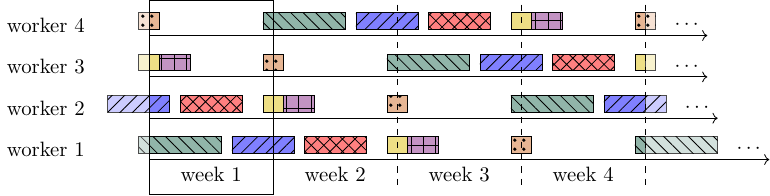}
    \caption{The assignment $f^\star$ of Figure~\ref{fig:algo_result} labels arcs of a Hamiltonian cycle in $D^{\calI', \overline\calF_{\calI'}}$. A periodic balanced feasible assignment with period $q$ can be built from the first week of assignment $f^\star$, as illustrated here.}
    \label{fig:hamiltonian_balanced}
\end{figure} 

\section{A problem of pebbles on an arc-colored Eulerian directed graph}\label{sec:pebbles}

This section introduces a problem of pebbles moving on an Eulerian directed graph, which we believe to be interesting for its own sake. The proof of Theorems~\ref{thm:periodic} and~\ref{thm:extension_cns} essentially consists in reducing the problem of existence of a balanced feasible assignment to this pebble problem. From now on, this section does not refer anymore to the question of assignments and periodic tasks.

Consider an arc-colored Eulerian directed multigraph $D=(V,A)$ such that each vertex is the head of an arc of every color, and also the tail of an arc of every color. (In other words, each color is a collection of vertex disjoint directed cycles covering the vertex set.) Assume we have a pebble on each vertex. We denote by $P$ the set of pebbles, and we have thus $|P| = |V|$.

Now, we explain how a sequence of colors induces a sequence of moves for the pebbles. Given a sequence $c_1,c_2,\ldots$ of colors, each pebble is first moved along the unique arc of color $c_1$ leaving the vertex on which it is originally located; then it is moved along the unique arc of color $c_2$ leaving the vertex it has reached after the first move; and so on. Remark that each move sends each pebble on a distinct vertex and so after each move, there is again a pebble on each vertex.

An illustration of such an arc-colored Eulerian directed multigraph and a move of the pebbles is given in Figure~\ref{fig:pebbles}.

\begin{figure}
    \centering
    \begin{subfigure}{0.45\textwidth}
        \centering
            \includegraphics[width = 0.9\textwidth]{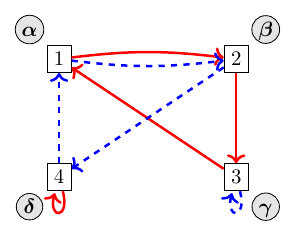}
        \caption{Graph $D$ with four vertices $V=\{1,2,3,4\}$ and pebbles $P = \{\boldsymbol{\alpha}, \boldsymbol{\beta}, \boldsymbol{\gamma}, \boldsymbol{\delta}\}$.}
    \end{subfigure}
    \hfill
    \begin{subfigure}{0.45\textwidth}
        \centering
            \includegraphics[width = 0.9\textwidth]{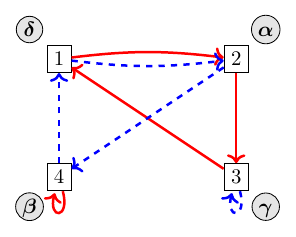}
        \caption{Pebbles on the graph $D$ after moving following the color dashed blue.}
    \end{subfigure}
    
    \caption{Example of an arc-colored Eulerian directed multigraph.}
    \label{fig:pebbles}
\end{figure}

We might ask under which condition there exists an infinite sequence of colors such that the arc visits are ``balanced,'' i.e., each pebble visits each arc with the same frequency. Not only such a sequence always exists but such a sequence can be chosen to be periodic.

\begin{proposition}\label{prop:graph}
There always exists a periodic sequence of colors making each pebble visit each arc with the same frequency $1/|A|$.
\end{proposition}

The proof shows a bit more: each pebble actually follows a periodic walk on $D$ which has the same period as the sequence of colors, and the latter is upper bounded by $|A|(|V|-1)!$.

The proof of this proposition relies on a larger graph $\widetilde{D} = (\widetilde{V}, \widetilde{A})$ built as follows. The vertex set $\widetilde{V}$ is the set of bijections from $P$ to $V$. For every color $c$, define the permutation $\sigma_c$ of $V$ by setting $\sigma_c(i) = i'$ whenever there is an arc of color $c$ from $i$ to $i'$ in $D$. The set $\widetilde{A}$ is built as follows: for each bijection $\eta\colon P \to V$ and each color $c$, introduce an arc $(\eta, \sigma_c\circ \eta)$, and color this arc with color $c$. The indegree and outdegree of every vertex in $\widetilde{V}$ are equal to the number of colors.

For each pebble $j$, we introduce a function $p_j \colon \widetilde{A}\to A$. Given an arc $\tilde{a}=(\eta, \eta')$ of $\widetilde{A}$ with color $c$, we define $p_j(\tilde{a})$ as the arc $(\eta(j), \eta'(j))$ of $A$ with color $c$.

The graph $\widetilde{D}$ is an encoding of all possible distributions of the pebbles on $V$ and all possible transitions between these distributions. More precisely consider any initial distribution $\eta$ of the pebbles on $V$ and a sequence of colors $c_1,c_2,\ldots$. The moves induced by the sequence of colors translate into a walk on $\widetilde{D}$. The corresponding sequence of vertices of $\widetilde{D}$ is the sequence of distributions of the pebbles on $V$ induced by the sequence of colors.

\begin{lemma}\label{lem:factorial}
    Let $j\in P$ and $a\in A$. Denoting by $\kappa$ the number of connected components of $\widetilde{D}$, we have 
    \[
    |p_j^{-1}(a)\cap A(\widetilde{K})| = \frac{(|V|-1)!} \kappa
    \]
    for every connected component $\widetilde{K}$ of $\widetilde{D}$. In particular, the left-hand term is independent of $j$, $a$, and $\widetilde K$.
\end{lemma}

Note that weakly and strongly connected components of $\widetilde{D}$ are identical by equality of the in- and outdegrees.

\begin{proof}[Proof of Lemma~\ref{lem:factorial}]
    Denote by $c$ the color of $a$.
    
    We prove first that every connected component $\widetilde{K}$ of $\widetilde{D}$ contains at least one arc from $p_j^{-1}(a)$. Let $\eta$ be a vertex of such a connected component $\widetilde{K}$. Consider any walk $W$ in $D$ from $\eta(j)$ to the tail of $a$, and then traversing $a$. Such a walk exists because $D$ is strongly connected. With $c_1,c_2,\ldots,c$ being the sequence of colors of the arcs traversed by the walk, the sequence 
    \[
\eta,\sigma_{c_1}\circ\eta,\sigma_{c_2}\circ\sigma_{c_1}\circ\eta,\ldots,\sigma_c\circ\cdots\circ\sigma_{c_2}\circ\sigma_{c_1}\circ\eta
\]
forms a walk in $\widetilde{K}$ starting from $\eta$, whose image by $p_j$ is $W$. Hence, $\widetilde{K}$ contains at least one arc from $p_j^{-1}(a)$.

Second, given two components $\widetilde{K}_1$ and $\widetilde{K}_2$ of $\widetilde{D}$, we build an injective map $\psi \colon A(\widetilde{K}_1) \to \widetilde A$ as follows. Pick $\tilde a_1 \in p_j^{-1}(a) \cap A(\widetilde{K}_1)$ and $\tilde a_2 \in p_j^{-1}(a) \cap A(\widetilde{K}_2)$. According to what we have just proved, these two arcs exist. Write $\tilde a_1 = (\eta_1,\sigma_c \circ \eta_1)$ and $\tilde a_2 = (\eta_2,\sigma_c \circ \eta_2)$. Then, for an arc $\tilde a \in A(\widetilde{K}_1)$ with tail vertex $\eta$ and color $d$, set $\psi(\tilde a)$ as the arc $(\eta \circ \eta_1^{-1} \circ \eta_2, \sigma_d \circ \eta \circ \eta_1^{-1} \circ \eta_2)$ with color $d$ (this arc is unique). Checking that $\psi$ is injective is immediate.

Third, we check that $\psi$ maps elements from $p_j^{-1}(a) \cap A(\widetilde{K}_1)$ to $p_j^{-1}(a) \cap A(\widetilde{K}_2)$. Let $\tilde a$ be an arc in $p_j^{-1}(a) \cap A(\widetilde{K}_1)$. It is of the form $(\eta,\sigma_c\circ\eta)$. Its image by $\psi$ is the arc $(\eta \circ \eta_1^{-1} \circ \eta_2, \sigma_c \circ \eta \circ \eta_1^{-1} \circ \eta_2)$ with color $c$. Denoting $i$ the tail of $a$, we have $\eta(j) = \eta_1(j) = \eta_2(j) = i$, which implies immediately that $p_j\bigl(\psi(\tilde a)\bigl)$ has the same endpoints as $a$. Since it has also the same color $c$, we have $p_j\bigl(\psi(\tilde a)\bigl) = a$.

From the previous two paragraphs, we see that for any two components $\widetilde{K}_1$ and $\widetilde{K}_2$ of $\widetilde{D}$, we have $|p_j^{-1}(a) \cap A(\widetilde{K}_1)| \leq |p_j^{-1}(a) \cap A(\widetilde{K}_2)|$.
Since the choices of $\widetilde{K}_1$ and $\widetilde{K}_2$ can be arbitrary, we have actually 
\begin{equation}\label{eq:equal}
|p_j^{-1}(a) \cap A(\widetilde{K}_1)| = |p_j^{-1}(a) \cap A(\widetilde{K}_2)|\, .
\end{equation}

Finally, an arc $\tilde a = (\eta,\eta')$ is mapped to $a$ by $p_j$ precisely when $\tilde a$ is colored with color $c$ and $\eta(j) = i$ (where $i$ is the tail of $a$). The number of bijections $\eta$ from $P$ to $V$ with $\eta(j)=i$ is $(|V|-1)!$. Hence, $|p_j^{-1}(a)| = (|V|-1)!$. Combining this with equality~\eqref{eq:equal}, we get the desired conclusion.
\end{proof}

\begin{proof}[Proof of Proposition~\ref{prop:graph}]
    Choose any connected component $\widetilde{K}$ of $\widetilde{D}$. It is Eulerian, since each vertex of $\widetilde{D}$ has equal in- and outdegrees. Consider an arbitrary Eulerian cycle, and denote by $c_1,c_2,\ldots$ the sequence of colors of the arcs of this cycle. According to Lemma~\ref{lem:factorial}, every pebble $j$ moved according to this sequence of colors follows a closed walk on $D$ visiting each arc $\frac{(|V|-1)!} \kappa$ times. Repeating infinitely many times this sequence of colors provides the desired periodic sequence.
\end{proof}

\section{Proofs of Theorems~\texorpdfstring{\ref{thm:periodic}}{1.3} and~\texorpdfstring{\ref{thm:extension_cns}}{1.4}}
The following proofs and results are established for the extended version. They are also valid for the basic version, as it is a special case of the extended version, even if Theorems~\ref{thm:cns} and \ref{thm:algo} provide us with stronger results.
\begin{lemma}\label{lem:eulerian_periodic}
    Let $\calI$ be an instance such that $|U(\calI)|=q$ and $\calF$ be an arbitrary set of feasible assignments.
    If $D^{\calI, \calF}$ is Eulerian, then there exists a periodic balanced feasible assignment for $\calI$ with period bounded by $|\calF|q!$.
\end{lemma}

\begin{proof}
    Suppose that $D^{\calI, \calF}$ is Eulerian. Locate one pebble on each vertex of $D^{\calI,\calF}$. Applying Proposition~\ref{prop:graph} on $D^{\calI,\calF}$, with each feasible assignment in $D^{\calI,\calF}$ identified with a color, we get a periodic sequence $f_1,f_2,\ldots$ of feasible assignments. Denote by $g$ the resulting feasible assignment given by Lemma~\ref{lem:seq} for this sequence, with $\pi_1$ being an arbitrary permutation.

    Number $j = g(i,1)$ the pebble initially located on vertex $i\in U(\calI)$. This makes sure that each pebble gets a distinct number in $[q]$ (by the bijectivity of $\varphi_{g,1}$). We establish now the following claim: {\em For every $i\in U(\calI)$ and every $j\in [q]$, pebble $j$ is on vertex $i$ after its $r$th move if and only if $g(i,r+1)=j$.} %{\blue The pebble is said to be on a vertex after its $0$th move if it is located on this vertex initially.}

    Let us proceed by induction on $r\in \Z_{\geq0}$. This is true for $r=0$ by the definition of the numbering of the pebbles. Suppose now that the claim if true for some $r\in\Z_{\geq 0}$. Consider pebble $j$ and assume it is located on $i$ after its $r+1$th move. This means that the pebble $j$ was on vertex $i'$ after its $r$th move then moved along the arc from $i'$ to $i$ labeled with $f_{r+1}$. Then, using equation~\eqref{eq:pi} and the fact that $f_{r+1}(i,2) = f_{r+1}(i',1)$ by definition of $D^{\calI, \calF}$,
    \begin{equation}\label{eq:pebble}
    g(i,r+2) = \pi_{r+2}(f_{r+2}(i,1)) = \pi_{r+1}(f_{r+1}(i,2)) = \pi_{r+1}(f_{r+1}(i',1)) = g(i',r+1)=j\, ,
    \end{equation}
    as desired. 
    Conversely, assume that $g(i,r+2) = j$. Denote by $i'$ the tail of the arc of head $i$ and label $f_{r+1}$. Then equation~\eqref{eq:pebble} holds as well, meaning that $g(i',r+1)=j$. By induction, the pebble $j$ was located on vertex $i'$ after its $r$th move. It then moves along the arc from $i'$ to $i$ with label $f_{r+1}$, which concludes the proof of the claim.
    
    We check that $g$ is balanced. According to Proposition~\ref{prop:graph}, for every $i \in U(\calI)$, every $j \in [q]$, and every $f \in \calF$, we have
    \[
    \lim_{t \to +\infty}\frac 1 t\big|\{r \in [t]\colon \text{pebble }j\text{ arrives at }i\text{ along arc labeled with } f \text{ for its $r$th move}\}\big| = \frac 1 {|A^{\calF}|}\, .
    \]
    %(After its $0$th move, a pebble is at its initial vertex. This does actually not really matter since it is an asymptotic result.)

    With the claim, this equality becomes
    \[
    \lim_{t \to +\infty}\frac 1 t\big|\{r \in [t] \colon f_{r} = f \text{ and } g(i,r+1)=j \}\big| = \frac 1 {|A^{\calF}|}\, .
    \]
    This equality is actually also true when $U(\calI)$ is replaced by the larger set $[n]$. Indeed, we have $g(i,r+1)=g(u_{f_{r+1},1}(i),r+1)$ for every $r \in \Z_{\geq 0}$, by definition of $g$ and $u_{f_{r+1},1}(i)$ (the latter being defined in Section~\ref{sec:thm1}). 
    
    Therefore, for every $i\in [n]$ and every $j\in [q]$, we have
    \begin{align*}
    \lim_{t\rightarrow +\infty} \frac{1}{t}\big|\{r \in [t] \colon g(i,r+1) = j\}\big| &= \sum_{f\in \calF}\lim_{t \to +\infty}\frac 1 t\big|\{r \in [t] \colon f_{r} = f \text{ and } g(i,r+1)=j \}\big| \\ 
    &= \frac{|\calF|}{|A^{\calF}|} = \frac{1}{q} \, .
    \end{align*} 
    This leads to the desired relation $\lim_{t\rightarrow +\infty} \frac{1}{t}\big|\{r \in [t] \colon g(i,r) = j\}\big|$ since the element $r=1$ does not count asymptotically.

    We check that $g$ is periodic with period upper bounded by $|\calF|q!$.
    According to the comment following Proposition~\ref{prop:graph}, each pebble follows a periodic walk. Denote by $h$ the period of this walk. Still according to the same comment, we have $h\leq |A^\calF|(q-1)! = |\calF|q!$ and that the sequence of the $f_r$ is periodic with period $h$. With the claim proven above, for each $i\in U(\calI)$ and $r\in \Z_{>0}$, we thus have $g(i,r)=g(i,r+h)$. This relation is true for all $i\in[n]$ because
    
    \begin{align*}
        g(i,r)&=\pi_r(f_r(i,1)) = \pi_r(f_r(u_{f_r,1}(i),1)) = g(u_{f_r,1}(i),r) = g(u_{f_r,1}(i),r+h) \\&= \pi_{r+h}(f_{r+h}(u_{f_{r},1}(i),1)) = \pi_{r+h}(f_{r+h}(u_{f_{r+h},1}(i),1)) = \pi_{r+h}(f_{r+h}(i,1)) = g(i,r+h)\, .
    \end{align*}
    Therefore, the assignment $g$ is periodic with period $h$.
\end{proof}

\begin{proof}[Proof of Theorem~\ref{thm:periodic}]
    Denote by $\calI'$ the new instance defined from $\calI$ as in Section~\ref{sec:transfo} and take $\calF_{\calI'}$ a universal set of feasible assignments for $\calI'$ with size bounded by $q^2$. Such a set exists, because there are at most $q^2$ pairs of elements in $U(\calI')$.
    Suppose there exists a balanced feasible assignment for $\calI$. According to Lemma~\ref{lem:connected_balanced}, the graph $D^{\calI',\calF_{\calI'}}$ is Eulerian. According to Lemma~\ref{lem:eulerian_periodic}, there exists a periodic balanced feasible assignment for $\calI'$ with period bounded by ${|\calF_{\calI'}| q!\leq} q^2q!$. Therefore there exists such an assignment for $\calI$ (by Lemma~\ref{lem:transfo_balanced}).
\end{proof}

\begin{proof}[Proof of Theorem~\ref{thm:extension_cns}]
    As discussed in Section~\ref{sec:extension}, we just have to prove the implication $\ref{ext_cond1}\Rightarrow\ref{ext_cond2}$.
    Suppose there exists a feasible assignment for $\calI$ with a worker performing each task at least once.
    Take an arbitrary universal set $\calF$ of feasible assignments. According to Lemma~\ref{lem:induced_walk}, the graph $D^{\calI,\calF}$ is weakly connected, thus Eulerian. According to Lemma~\ref{lem:eulerian_periodic}, there exists a balanced feasible assignment for $\calI$.
\end{proof}

\section{Concluding remarks}
\subsection{All feasible assignments are balanced. (Almost.)}\label{subsec:all}

If we are just interested in the existence of a balanced feasible assignment, and not on the periodicity of such an assignment or its computability, we can replace Proposition~\ref{prop:graph} by the following lemma in the proof of Theorem~\ref{thm:extension_cns}. We keep the same setting of an arc-colored Eulerian directed multi-graph $D=(V,A)$ with a distribution of pebbles on its vertices, as in the beginning of Section~\ref{sec:pebbles}.

\begin{lemma}\label{lem:graph-bis}
    Consider an infinite sequence of independent random colors drawn uniformly. Then, almost surely, this sequence makes each pebble visit each arc with the same frequency. 
\end{lemma}

In particular there are infinitely many color sequences making each pebble visit each arc with the same frequency. The proof relies on basic properties of Markov chains. (A standard reference on Markov chains is the book by Norris~\cite{norris1998markov}.) The proof does not show how to construct such a sequence of colors. It is not even clear that the proof could be modified in that regard. So the proof shows that almost all color sequences have the desired property, but does not explain how to construct a single such sequence. Although this might sound surprising, this phenomenon is quite common. {\em Normal} numbers form an example: (almost) all numbers are normal but not a single one has been described explicitely~\cite{nomal}.

The proof of Proposition~\ref{prop:graph} actually provides an alternative proof of the existence of sequences of colors making each pebble visit each arc with the same frequency, with an explicit construction. However, the latter proof does not show that almost all sequences are actually like that. (In counterpart, it shows that such a sequence can be chosen to be periodic.)

\begin{proof}[Proof of Lemma~\ref{lem:graph-bis}]
  Any realization of this random sequence of colors defines a sequence of moves for the pebbles, as described above. Consider an arbitrary pebble. The random sequence of colors translates thus into a random walk of the pebble on the graph $D$. Denote by $X_k$ the arc along which the pebble performs its $k$th move. The $X_k$'s form a finite Markov chain. Since the graph is Eulerian, this Markov chain is irreducible, and hence there exists a unique invariant distribution $\llambda$ such that $\llambda^{\top} = \llambda^{\top} M$, where $M$ is the transition matrix of the Markov chain.

    We claim that $\llambda$ is actually the vector $\frac 1 {|A|} \un$, where $\un$ is the all-one vector. By the uniqueness of the invariant distribution, it is enough to check that $\un$ is a left eigenvector of $M$ with eigenvalue equal to $1$. The entry $M_{a,a'}$ of the transition matrix (row $a$, column $a'$), which corresponds to the probability of moving along $a'$ just after moving along $a$, is equal to $1/\ell$ if the head of $a$ is the tail of $a'$, and $0$ otherwise. The indegree of each vertex being $\ell$, we have 
    \[
    \sum_{a \in A}M_{a,a'} = \ell \frac 1 \ell = 1 \, ,
    \]
    and therefore $\un^{\top} M = \un^{\top}$.

According to the ergodic theorem, for almost all realizations of the random sequence of colors, the pebble visits any arc $a$ with a frequency equal to the corresponding entry in $\llambda$, which is equal to $\frac 1 {|A|}$ for all arcs since it is a probability distribution proportional to the all-one vector. The previous discussion does not depend on the considered pebble, which implies the desired result: for almost all realizations of the random sequence of colors, every pebble visits every arc with a frequency equal to $\frac 1 {|A|}$.
\end{proof}

% EST-CE QUE C EST TJRS VRAI ???
%Similarly to the proof of Theorem~\ref{thm:cns}, using a much larger set of feasible assignments $\calF'$ (typically one for each feasible ``pattern'' on $[1,2]$), Lemma~\ref{lem:graph-bis} could translate into the following statement:
%{\em If there is at least one balanced feasible assignment, then almost all feasible assignments are balanced.}

\subsection{Minimizing the number of workers}
A follow-up problem could be finding the minimum number of workers such that there exists a balanced feasible assignment (which is actually equivalent to the existence of a periodic balanced feasible assignment according to Theorem~\ref{thm:periodic}). Within the setting of the basic version, namely without schedules, the problem is polynomial thanks to Theorem~\ref{thm:algo} and Lemma~\ref{lem:polynomial_bigq}. Indeed, we can use Theorem~\ref{thm:algo} for each value of $q$ within $[0,2n]$, the bound $2n$ being given by Lemma~\ref{lem:polynomial_bigq}.
(The existence of a balanced feasible assignment with $q$ workers implies the existence of such an assignment for $q+1$ workers: indeed, there is still a worker performing each task at least once and Theorem~\ref{thm:cns} applies; thus the complexity of this algorithm can actually be improved by performing a binary search.)

\subsection{Open question}
For the extended version, we do not know whether there is a set of valid schedules for which the bound on the period of a periodic balanced feasible assignment obtained in Theorem~\ref{thm:periodic} is tight.

\bibliographystyle{amsplain}
\bibliography{bibliography}

\end{document}